\documentclass[12pt, draftclsnofoot, onecolumn]{IEEEtran}
\usepackage[utf8]{inputenc}
\usepackage[T1]{fontenc}
\usepackage{amssymb,amsmath}
\usepackage{cite}
\usepackage{psfrag}
\usepackage{url}
\usepackage[absolute,overlay]{textpos}

\usepackage{pgf}
\usepackage[export]{adjustbox}
\usepackage{bm}
\usepackage{bbm}
\usepackage{booktabs}
\usepackage{url}
\usepackage{upgreek}
\usepackage{verbatimbox}
\usepackage{algorithm}
\usepackage{algorithmic}
\usepackage{enumitem}
\usepackage{caption}
\usepackage{amsthm}

\newtheorem{theorem}{Theorem}
\newtheorem{corollary}{Corollary}

\newcommand{\appropto}{\mathrel{\vcenter{
			\offinterlineskip\halign{\hfil$##$\cr
				\propto\cr\noalign{\kern2pt}\sim\cr\noalign{\kern-2pt}}}}}
\usepackage{amsthm}

\usepackage{colortbl}

\begin{document}
	\title{Coordinated Pilot Transmissions for Detecting the   Signal Sparsity Level in a Massive IoT Network    under Rayleigh Fading}		
	\author{Onel L. A. López,
	Glauber Brante,
	Richard D. Souza,
	Markku Juntti,
	and Matti Latva-aho
	\thanks{O. López, M. Juntti, and M. Latva-aho are  with  the Centre for Wireless Communications University of Oulu, Finland, e-mails: \{Onel.AlcarazLopez, Markku.Juntti, Matti.Latva-aho\}@oulu.fi.
		G. Brante is with the Department of Electrotechnics, Federal University of Technology PR, Curitiba, Brazil, e-mail: gbrante@utfpr.edu.br.
		R. D. Souza is with the Electrical and Electronics Engineering Department,
		Federal University of Santa Catarina (UFSC), Florianópolis, Brazil, e-mail:
		richard.demo@ufsc.br.}	\thanks{This research has been supported by the Academy of Finland, 6G Flagship program under Grant 346208, and Brazilian National Council for Scientific and Technological Development (CNPq).} 	} 	

	\maketitle
	\vspace{-12mm}
\begin{abstract}
    \vspace{-2mm}
    Grant-free protocols exploiting compressed sensing (CS) multi-user detection (MUD) are appealing for solving the random access problem in massive machine-type communications (mMTC) with sporadic device activity. 
    Such protocols would greatly benefit from a prior deterministic knowledge of the sparsity level, i.e., instantaneous number of simultaneously active devices $K$. Aiming at this, herein we introduce a framework relying on coordinated pilot transmissions (CPT) over a short phase at the beginning of the transmission block for detecting $K$ in mMTC scenarios under Rayleigh fading.
    CPT can be implemented either as: i) U-CPT, which exploits only uplink transmissions, or A-CPT, which includes also downlink transmissions for channel state information (CSI) acquisition that resolve fading uncertainty. We discuss two specific implementations of A-CPT: ii) A-CPT-F, which implements CSI-based phase corrections while leveraging the same statistical inverse power control used by U-CPT, and iii) A-CPT-D, which implements a dynamic CSI-based inverse power control, although it requires some active devices to remain in silence if their corresponding channels are too faded. 
    We derive a signal sparsity level detector for each CPT mechanism by relaxing the original integer detection/classification problem to an estimation problem in the continuous real domain followed by a rounding operation. We show that the variance of the relaxed estimator
    increases with $K^2$ and $K$ when operating with U-CPT and A-CPT mechanisms, respectively. The distribution of the estimators under U-CPT, A-CPT-F and A-CPT-D is found to follow an exponential, Gaussian, and Student's $t-$like distribution, respectively. Analyses evince the superiority of A-CPT-D, which is also corroborated via numerical results. We reveal several interesting trade-offs and highlight potential research directions.     
\end{abstract}
	%
%
\IEEEpeerreviewmaketitle
\vspace{-4mm}
\section{Introduction}\label{intro}
\vspace{-2mm}
Nowadays, the number of Internet of Things (IoT) devices  is exponentially growing 
driven by the need of turning our homes, vehicles, entertainment, health, work, industries, and social/community services into smart, autonomous, sustainable, interactive and intelligent environments \cite{Mahmood.2020,Choi.2022,MahmoodL.2020,Carvalho.2017,Lopez.2018}. 
As connectivity backbone of the IoT, massive machine-type communication (mMTC) paradigm aims to address the corresponding connectivity challenges. The latter are intertwined with the unique features of massive IoT connectivity setups, specifically \cite{Mahmood.2020,Choi.2022,MahmoodL.2020,Carvalho.2017}: 
i) sporadic transmissions, i.e., an unknown/random subset of machine-type communication devices, called simply devices in the sequel, is active at a given time instant; ii)  short-packet communications dominated by uplink (UL) traffic; and iii) energy-limited communication/operation. 
The third feature evinces the need of energy-efficient communication/operation protocols and, in many cases, battery-free operation \cite{Lopez.2020,Lopez.2021}; while all the features, in particular the first two,  call for novel multiple-access mechanisms \cite{Carvalho.2017,Lopez.2018,MahmoodL.2020}. Our work here contributes to the latter.

Grant-free (GF)  multiple-access protocols are particularly appealing for mMTC, since they \cite{Mahmood.2020,Choi.2022,MahmoodL.2020} i) promote efficient spectrum utilization as each device is not assigned a dedicated transmission resource block, ii) reduce signaling overhead, and iii)
improve energy-efficiency of the devices.
Note that owing to the network massiveness, it is impossible to assign orthogonal pilot sequences/preambles to the devices, thus, motivating the need of GF non-orthogonal multiple access protocols. 
However, a key challenge here lies in efficiently  identifying the set of sporadically active, non-orthogonally coexisting, devices and their data,
 for which 
collision resolution
mechanisms are required \cite{MahmoodL.2020,Carvalho.2017}.

We can distinguish two basic types of collisions: \textit{hard} and \textit{soft} collisions. The former occurs when  exactly the same preamble is being simultaneously used by several active devices, while the latter occurs when active devices use different non-orthogonal preambles as they interfere to some extent among each other. The probability of hard/soft collisions increases/decreases as the number of available preambles reduces. Since hard collisions are difficult to resolve without 
relying on 
sufficiently orthogonal channel subspaces \cite{Yin.2013,Adhikary.2013,Ribeiro.2020} and/or additional communication overhead, 
increasing the pool of non-orthogonal preambles (thus, favoring the occurrences of soft instead of hard collisions) is usually recommended in practice \cite{Choi.2022,MahmoodL.2020,Carvalho.2017}.  
A promising class of soft collision resolution methods, known as compressed sensing (CS) techniques, have been considered for multi-user device detection (MUD) in mMTC \cite{Choi.2017}.
\vspace{-4mm}
\subsection{Related Work}
\vspace{-2mm}
CS-MUD solutions usually rely on regularization, greedy, message-passing (MP), and/or 
artificial intelligence (AI) techniques. 
\begin{enumerate}[wide, labelwidth=!]
	\item \emph{Regularized MUD}
 relies on transforming the highly non-convex CS-MUD problem to convex via regularization and iterative procedures. For instance, Zhu and Giannakis \cite{Zhu.2011} proposed a ridge detector and a least absolute shrinkage and selection operator (LASSO) detector, which directly regularize the original CS-MUD problem based on $l_2-$ and $l_1-$norm, respectively. Later on, some sparsity-aware successive interference cancellation
 regularization techniques were proposed in \cite{Knoop.2013,Ahn.2018} aiming at lowering the detection complexity by sequentially recovering the transmitted symbols.
Meanwhile, Renna and Lamare \cite{Renna.2019} incorporated an $l_1-$norm regularization into an iteratively updated linear minimum mean square error (MMSE) filter and a constellation-list scheme to enable sparse detection. 
 Moreover, a joint user identification and channel estimation approach exploiting the alternating direction method of multipliers 
  was proposed in 
  \cite{Djelouat.2021}.
%
\item \emph{Greedy MUD}
has low complexity and generally only requires appropriate tuning of the termination of the transmitted signal/vector reconstruction. Schepker and Dekorsy \cite{Schepker.2011} applied for the first time
the 
 orthogonal least squares (OLS) and orthogonal matching pursuit (OMP) greedy algorithms to a  sparse
mMTC scenario. 
Since the latter 
outperforms  the 
former,
 latest research on greedy MUD has mostly focused on 
OMP-based algorithms. 
For instance,
Schepker \textit{et al.} \cite{Schepker.2015}
 proposed group OMP (GOMP) leveraging channel decoders for greater performance, while
Xiong \textit{et al.} \cite{Xiong.2014} proposed a detection-based orthogonal matching pursuit (DOMP) algorithm which, unlike the conventional OMP, does not rely on the priori knowledge of the signal/device sparsity (the number of active devices). Specifically, DOMP runs
binary hypothesis on the residual vector of OMP at each iteration, while stopping when there is no signal component in the residual vector.
%
Meanwhile, 
a noise-robust greedy algorithm exploiting \textit{a posteriori} probability ratios for every index of sparse input signals is designed in \cite{Yu0.2019}. 
Finally, Lee and Yu \cite{Kyubihn.2020} leveraged prior activation probability of each device to improve the performance of several greedy MUD schemes in mMTC, and show that they are robust against prior information inaccuracy.
%
\item \emph{MP-based MUD}
constitutes a class of algorithms exploiting factor graphs,
 thus,
 the \textit{a posteriori} distribution of the signal to be reconstructed. In practice, due to the large-scale nature of the access problem in mMTC, the usual approach is to adopt/design approximate MP (AMP) algorithms relying on  iterative thresholding, which also allows analytic performance characterization via the so-called state evolution.
For instance,  Chen \textit{et al.} \cite{Chen.2018} 
derived 
MMSE denoisers for AMP depending on whether or not the large-scale component of the channel fading is known. 
Senel and Larsson \cite{Kamil.2018} analyzed and proposed algorithmic enhancements for coherent and non-coherent MUD based on AMP.  Meanwhile,
Ke \textit{et al.} \cite{Ke.2020} designed non-orthogonal pseudo-random pilots for UL broadband massive access. They formulated the active user detection and channel estimation as a generalized multiple measurement vector CS problem and solved it via
a  generalized multiple measurement vector AMP  algorithm.
The suitability of AMP for joint device activity detection and channel estimation of devices coexisting with enhanced mobile broadband (eMBB) services is assessed and promoted in \cite{Marata.2021}. 
Finally, 
Wang \textit{et al.} \cite{Wang.2016} designed an
 AMP algorithm that exploits the  temporal activation correlation of the devices, and showed the achievable performance gains.
%
\item \emph{AI-based MUD} leads to direct  detection decisions as the detection parameters are learned and configured on the go, thus, avoiding empirical parameters tuning.
Deep learning 
 is the most commonly used  AI tool for solving the CS-MUD problem \cite{Ye.2021}. 
Some examples of MUD based on deep learning can be encountered in \cite{Bai.2019,Cui.2021,Chen.2021}. Specifically, Bai \textit{et al.} \cite{Bai.2019} proposed a fast data-driven 
 algorithm for CS-MUD in mMTC relying on a novel block restrictive activation nonlinear unit that nicely captures system sparsity. Meanwhile, Cui \textit{et al.} \cite{Cui.2021} designed two model-driven 
  approaches, which effectively utilize features of sparsity patterns in
designing common measurement matrices and adjusting state-of-the-art detectors/decoders.  Interestingly, the optimum depth (number of layers) to be configured in a deep neural network 
varies according to the sparsity statistics, which motivated the work in  \cite{Chen.2021}. Therein, a proposal to autonomously/dynamically update the number of layers in the inference phase is proposed, for which authors introduced an extra halting score at each layer.
\end{enumerate}
\vspace{-4mm}
\subsection{Contributions}
\vspace{-2mm}
State-of-the-art research on CS-MUD either assumes that i) signal sparsity level is known  and exploited for MUD, e.g., \cite{Zhu.2011,Knoop.2013,Ahn.2018,Renna.2019,Djelouat.2021},
 or ii) signal sparsity level detection is a sub-product of MUD, e.g., \cite{Schepker.2011,Schepker.2015,Xiong.2014,Yu0.2019,Kyubihn.2020,Chen.2018,Ke.2020,Marata.2021,Bai.2019,Cui.2021,Chen.2021}. 
 In the first case, there has been no insights/answers on how the sparsity level could be known in advance to MUD, which makes the mechanisms proposed in \cite{Zhu.2011,Knoop.2013,Ahn.2018,Renna.2019,Djelouat.2021} 
 strictly impractical so far. In the second case, the sparsity level information is not required. However, having and exploiting such information would  certainly improve the MUD performance. Specifically, the iterative mechanisms proposed in \cite{Schepker.2011,Schepker.2015,Xiong.2014,Yu0.2019,Kyubihn.2020,Chen.2018,Ke.2020,Marata.2021}
  face the challenge of  setting an appropriate stopping criterion, while the deep learning based mechanisms proposed in \cite{Bai.2019,Cui.2021} have fixed depth in terms of the number of layers and do not adapt well to highly-varying sparsity-levels. Interestingly, the depth could also be learned \cite{Chen.2021}, but this introduces a further non-linearity into the system. Although the approach leads to accuracy improvements with respect to state-of-the-art MUD based on deep learning, it is also more complex. From the discussion above, we can conclude that the CS-MUD mechanisms can all significantly benefit from a (sufficiently) deterministic prior on the sparsity level, which is specifically our aim here.
The information on the sparsity level enables the application of the MUD solutions in \cite{Zhu.2011,Knoop.2013,Ahn.2018,Renna.2019,Djelouat.2021},
 while it potentially makes those ones in \cite{Schepker.2011,Schepker.2015,Xiong.2014,Yu0.2019,Kyubihn.2020,Chen.2018,Ke.2020,Marata.2021,Bai.2019,Cui.2021,Chen.2021} more easily configurable and accurate.

%
We consider a mMTC deployment under quasi-static Rayleigh fading where a random set of $K$ devices become active. 
Our main contributions are six-fold:
\begin{itemize}[leftmargin=0cm]
	\item We introduce a framework for detecting\footnote{By convention \cite{Kay.1993,Kay.1998}, a detection or classification operation is applied over a (discrete) set of possible hypotheses, while an estimation operation is not restricted to a discrete/natural domain. Hence, a detector for $K$ outputs an integer solution, while an estimator for $K$ may output a real solution.} $K$ by exploiting $N$ symbols at the beginning of a transmission block. The framework relies on coordinated pilot transmissions (CPT), and may  play  a crucial role for sparse signal recovery MUD algorithms in mMTC.
	\item We propose unassisted CPT (U-CPT) and  assisted CPT (A-CPT) mechanisms.  Specifically, only UL transmissions are exploited when using U-CPT, while A-CPT mechanisms include also downlink (DL) transmissions for CSI estimation that resolve fading uncertainty. We discuss two A-CPT specific implementations: i) A-CPT-F, which implements CSI-based phase corrections while leveraging the same statistical inverse power control used by U-CPT, and ii) A-CPT-D, which implements a dynamic CSI-based inverse power control, although it requires that the active devices remain in silence with probability $\xi$ given an average transmit power constraint.
	\item 
	We derive the optimum detector of the value of $K$ in the case of U-CPT. Meanwhile, since the optimum detection relies on a combinatorial search in the case of A-CPT mechanisms, we relax the problem to the continuous domain and find efficient estimators for $K$, also for U-CPT. The estimator variance is shown to increase with $K^2$ and $K$ when operating respectively with U-CPT and A-CPT mechanisms.
	\item Because of the relaxation mentioned above, the estimators need a quantization or rounding operation to find the integer value of $K$. We discuss two such approaches: 
	 i) 
	rounding to the nearest integer (NI), and ii) rounding relying on maximum likelihood (ML). We also numerically assess their performance.
	The results reveal that the NI-based rounding offers a performance similar to that of the more sophisticated ML-based rounding.
	\item We derive exact or approximate (semi) closed-form expressions for the probability mass function (PMF) of the detectors, and validate their accuracy. Such PMFs were shown to be exponential-, Gaussian- and Student's $t-$like in the case of U-CPT, A-CPT-F and A-CPT-D, respectively. They allow tractable computation of the estimation success probability, thus, becoming valuable for system design/optimization purposes.
	\item We evince the superiority of A-CPT-D under appropriate (but not strict/rigorous) configuration of $\xi$ (whose optimum value decreases slowly with $K$). Moreover, we show that the estimation accuracy increases with $N$, specially when operating with A-CPT-D, although the performance gain may saturate quickly in high SNR regimes.
\end{itemize}
Finally, we also identify and briefly discuss several attractive research directions related to CPT to pursue in the sequence.
\vspace{-4mm}
\subsection{Organization}
\vspace{-2mm}
The remainder of this paper is organized as follows. Section~\ref{system} introduces the system model and related problem.
Sections~\ref{design} and \ref{acpt} discuss the U-CPT and A-CPT mechanisms for detecting $K$, respectively, including  accuracy analyses.
The distribution of the (relaxed) estimators is characterized in Section~\ref{Dist}, while Section~\ref{results} presents and discusses numerical results. Finally, Section~\ref{conclusions}
concludes the article and highlights further research directions.

\textit{Notation:} 
Boldface lowercase letters denote column vectors. 
Superscripts $(\cdot)^*$ and $(\cdot)^H$ denote the complex conjugate and Hermitian operations, respectively. 
$||\cdot||$ is the Euclidean norm of a vector, $|\cdot|$ is the absolute (or cardinality for sets) operation, and $\mathrm{round}(\cdot)$ denotes integer rounding of the argument according to a certain rule. 
$\mathbb{P}(A)$ denotes the probability of the occurrence of event $A$, while $A|B$ denotes a random variable $A$ conditioned on $B$.
$\mathbb{E}[\!\ \cdot\ \!]$ and $\mathrm{var}[\!\ \cdot\!\ ]$ output expected value and variance of the argument, respectively.
%
 $\lfloor\cdot\rfloor$ and $\lceil\cdot\rceil$ are floor and ceiling functions, respectively, while $\binom{\ \!\cdot\!\ }{\!\ \cdot\!\ }$ is the binomial coefficient.   
$\Re\{\cdot\}$ ($\Im\{\cdot\}$) outputs the real (imaginary) part of the argument. Additionally, $\mathrm{li}(\ \!\cdot\ \!)$ is the logarithm integral \cite[eq. (6.2.8)]{Olver.2010}, $\mathrm{Ei}(\ \!\cdot\ \!)$ is the exponential integral \cite[eq. (6.2.1)]{Olver.2010},  $\mathrm{erfc}(\ \!\cdot\ \!)$ is the complementary error function \cite[eq. (7.2.2)]{Olver.2010}, $\Gamma(\ \!\cdot\ \!)$ is the complete gamma function \cite[eq. (5.2.1)]{Olver.2010}, and $K_\nu(\ \!\cdot\ \!)$ is the modified Bessel function of second kind and order $\nu$ \cite[eq. (10.27.4)]{Olver.2010}. Meanwhile, $B_{\nu}(\ \!\cdot\ \!,\ \!\cdot\ \!)$ and $I_{\nu}(\ \!\cdot\ \!,\ \!\cdot\ \!)$ with $\nu\in[0,1]$ are respectively the incomplete, and regularized incomplete, beta functions \cite[eq. (8.17.2)]{Olver.2010}.
$\mathbb{R}$ ($\mathbb{R}^+$) is the set of (non-negative) real numbers, $\mathbb{C}$ is the set of complex numbers, and  $\imath=\sqrt{-1}$ is the imaginary unit. $f_X(x)$ and $F_X(x)$ denote the probability density function (PDF) and cumulative distribution function (CDF), respectively, of a continuous random
variable $X$, while $p_Y(y)$ denotes the PMF of a discrete random variable $Y$. Finally, Table~\ref{table} lists the distributions used throughout the paper, including notations and main statistics.
\begin{table*}[!t]
	\centering
	\caption{Distributions used throughout the Paper and their Main Statistics}
	\vspace{-4mm}
	\begin{tabular}{p{2.7cm}llccccc}
		\toprule
		Distribution & Notation & Support	& PDF & CDF & Mean & Variance  \\ \midrule
		circularly-symmetric complex Gaussian &		$\mathcal{CN}(m,s^2)$ & $x\in \mathbb{C}$ & $e^{-|x-m|^2/s^2}/(\pi s)$ &  $-$ & $m$ & $s^2$   \\ 		
		Gaussian & $\mathcal{N}(m,s^2)$ & $x\in \mathbb{R}$ & $e^{-\frac{(x-m)^2}{2s^2}}/(\sqrt{2\pi}s)$ &  $1-\mathrm{erfc}\big(\frac{x-m}{\sqrt{2}s}\big)/2$ & $m$ & $s^2$   \\
		Exponential & $\mathrm{Exp}(\lambda)$ & $x\in\mathbb{R}^+$ & $\lambda e^{-\lambda x}$ & $1-e^{-\lambda x}$ & $1/\lambda$ & $1/\lambda^2$   \\
		Rayleigh & $\mathrm{Ray}(s)$ & $x\in\mathbb{R}^+$ & $xe^{-x^2/(2s^2)}/s^2$ & $1-e^{-x^2/(2s^2)}$ & $s\sqrt{\frac{\pi}{2}}$ & $(4-\pi)s^2/2$  \\
		Student's $t$ & $\mathcal{T}(\nu)$ & $x\in \mathbb{R}$ & $\frac{\Gamma(\frac{\nu+1}{2})}{\Gamma(\frac{\nu}{2})\sqrt{\nu\pi}}\big(1\!+\!\frac{x^2}{\nu}\big)^{-\frac{\nu+1}{2}}$ &  $1\!-\!\frac{1}{2}I_{\frac{\nu}{x^2\!+\!\nu}}\big(\frac{\nu}{2},\frac{1}{2}\big)$, $\forall x\!>\!0$ & $0$ & $\frac{s^2\nu}{\nu-2}$, $\forall \nu\!>\!2$   \\
		\bottomrule
	\end{tabular}\label{table}
	\vspace{-10mm}
\end{table*}
\vspace{-5mm}
\section{System \& Problem Description}\label{system}
\vspace{-2mm}
Consider a mMTC deployment, where a set $\mathcal{Q}$ of devices is served by a single coordinator, e.g., a BS or an aggregator. 
The MTC traffic is sporadic, i.e., only a random subset of the devices $\mathcal{K}\subseteq \mathcal{Q}$ is active at any given time. 
Assume that time is slotted and active devices must wait the next immediate time slot to start a (synchronous) transmission.
We aim at detecting the number of  devices $K=|\mathcal{K}|$, out of the total $Q=|\mathcal{Q}|$, becoming active in a given time slot, which is also referred to as signal sparsity level. This information is then potentially exploited in posterior detection/decoding mechanisms, e.g., \cite{Zhu.2011,Knoop.2013,Ahn.2018,Renna.2019,Djelouat.2021,Schepker.2011,Schepker.2015,Xiong.2014,Yu0.2019,Kyubihn.2020,Chen.2018,Ke.2020,Marata.2021,Bai.2019,Cui.2021,Chen.2021}. 

%
The devices and the coordinator are equipped with a single antenna.\footnote{As this is, to the best of our knowledge, the first work that proposes the sparsity level detection problem prior to a MUD phase, our aim here lies in introducing the basic ideas, principles, and performance baselines. The extension of our proposed mechanisms to multi-antenna setups are not only an interesting but a required research direction, which demands specific but non-trivial adjustments and analyses that we leave for a future work.}
Channels are subject to quasi-static Rayleigh fading, while DL and UL channels are reciprocal, which is motivated by the use of the same frequency band. 
Let us denote by $h_i$ the channel coefficient of the link between the coordinator and the $i-$th device, thus, 
$h_i\sim \ \mathcal{CN}(0,\beta_i)$, where $\beta_i$ is the average channel power gain. Moreover,
the coordinator and the devices are aware of the large-scale channel statistics, which is specially feasible in static and quasi-static MTC deployments.

\begin{figure}[t!]
	\centering
	\includegraphics[width=0.47\textwidth]{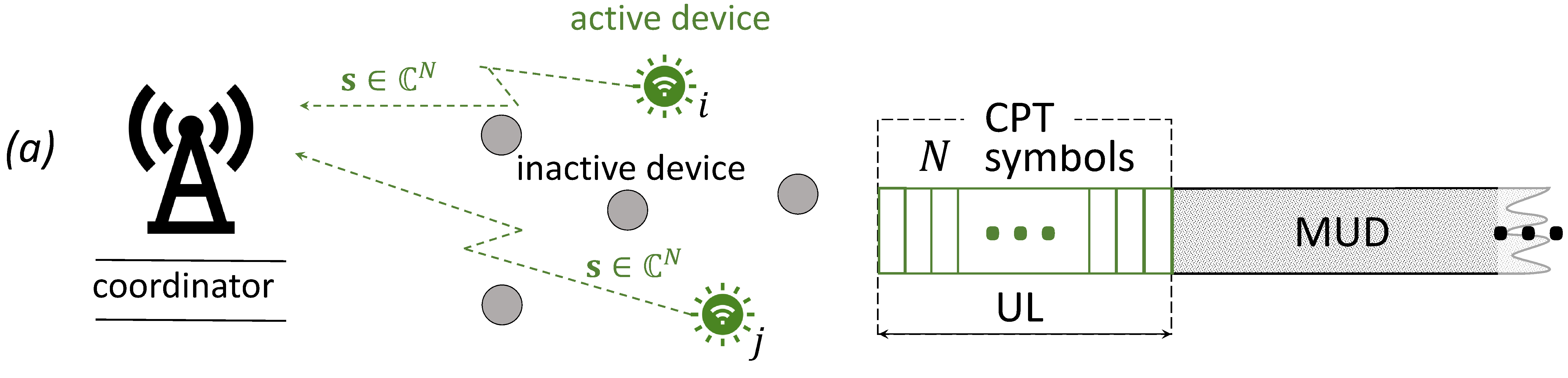}\ \ \ \ 	\includegraphics[width=0.47\textwidth]{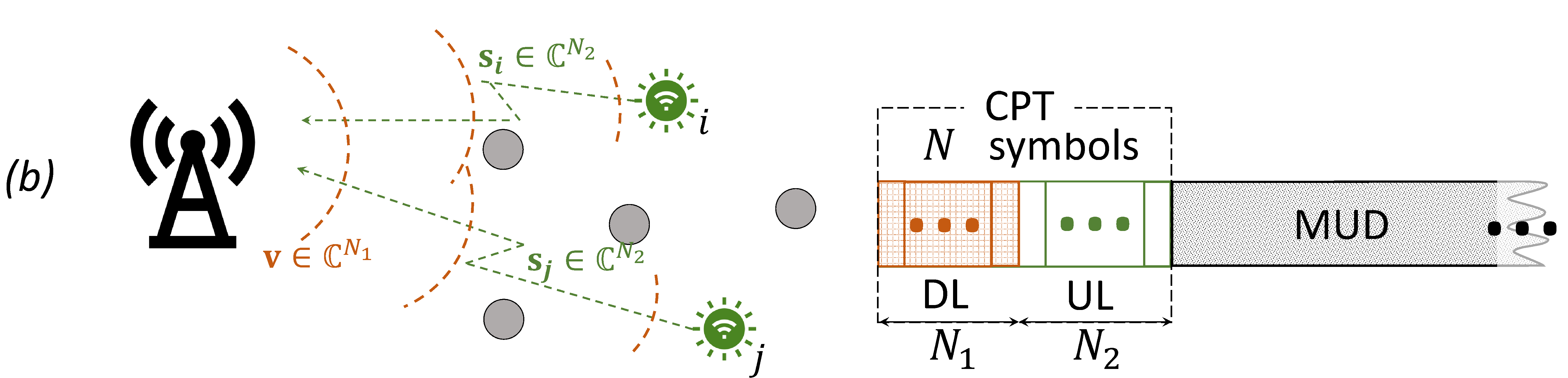}
	\vspace{-4mm}
	\caption{Proposed mechanisms: a) coordinator-unassisted CPT, and b) coordinator-assisted CPT.}
	\label{Fig1}
	\vspace{-10mm}
\end{figure}
Our proposed coordinated pilot transmission (CPT) mechanisms, which are schematically illustrated in Fig.~\ref{Fig1}, leverage the use of $N$ symbols for the purpose of detecting $K$. We assume that all the active devices transmit the same pilot sequence during such \emph{coordinated transmission phase}, either directly as in Fig.~\ref{Fig1}a, or after a short broadcast DL transmission from the coordinator  as in Fig.~\ref{Fig1}b. The signal $\mathbf{y}$ received at the coordinator 
is used to detect $K$ among all the $Q+1$ possible hypotheses:
%
	$\mathcal{H}_k:\ K=k$,
%
where $k=0,1,\cdots,Q$.
%
Let $f(\mathbf{y};\mathcal{H}_k)$ 
denote the PDF of $\mathbf{y}$
assuming that $\mathcal{H}_k$ is true. The optimal Neyman-Pearson (NP) detector, which models the true hypothesis as a non-random unknown, or the maximum {\it a posteriori} probability (MAP) decision rule with
 equal prior probabilities of the hypotheses 
 is given by \cite{Kay.1998} 
\begin{align}
	\hat{K}=\arg\max_{k} f(\mathbf{y};\mathcal{H}_k) \label{hK}.
\end{align}
Following the initial CPT phase, the transmission of training and data symbols in the case of coherent MUD, or only data symbols in the case of non-coherent MUD, occurs as traditionally. 
\vspace{-12mm}
\section{Unassisted CPT (U-CPT)}\label{design}
\vspace{-2mm}
%
Under the U-CPT approach, which is illustrated in Fig.~\ref{Fig1}a,
all the  devices that are active at the beginning of a given transmission block transmit immediately the corresponding shared pilot sequence $\mathbf{s}$ composed of $N$ complex symbols, i.e., $s[n]$, with $ n=1,\cdots,N$ such that $||\mathbf{s}||^2=N$.\footnote{For the purpose of security and resilience against malicious attacks, such sequence may belong to the orthogonal basis of all vectors of length $N$, and can be periodically and securely changed.}  Then, the $n-$th symbol received at the coordinator can be written as 
\begin{align}
y[n] &= \sum_{i\in\mathcal{K}}\nolimits\sqrt{p_i}s[n]h_i+w[n]=\sqrt{K\bar{p}}h's[n]+w[n],\ n=1,2,\cdots,N,
\end{align}
where 
$p_i$ is the per-symbol average transmit power of the $i-$th device, and $w[n]\in\mathcal{CN}(0,\sigma^2)$ denotes additive white Gaussian noise (AWGN) with variance $\sigma^2$. Moreover, last step is attained by adopting a statistical inverse power control with $p_i=\bar{p}/\beta_i$, where $\bar{p}$ is the target average receive power, and by setting $h'\triangleq h_i/\sqrt{\beta_i}\sim \mathcal{CN}(0,1)$.
%
%
\vspace{-6mm}
\subsection{Detector for $K$}\label{kucpt}
\vspace{-2mm}
%
%
\begin{theorem}\label{EK}
	Under U-CPT, the optimum detector \eqref{hK} is equivalent to \eqref{hKF} at the top of the next page, where
	%
		$	\alpha_1(\mathbf{x})\triangleq\sum_{n=1}^N x[n]^2/\sigma^2$, $
		\alpha_2(\mathbf{x})\triangleq 2\sqrt{k\bar{\gamma}}\sum_{n=1}^N x[n]\Im\{s[n]\}/\sigma$,
		$\alpha_3(\mathbf{x})\triangleq2\sqrt{k\bar{\gamma}}\sum_{n=1}^N x[n]\Re\{s[n]\}/\sigma$, 
		$\delta_1\triangleq 2k\bar{\gamma}\sum_{n=1}^{n}\Re\{s[n]\}\Im\{s[n]\}$, 
		$\delta_2\triangleq Nk\bar{\gamma}+1$,
	%
	and $\bar{\gamma}\triangleq \bar{p}/\sigma^2$ is the average receive signal to noise ratio (SNR).
\end{theorem}
\vspace{-3mm}
\begin{proof}
	See Appendix~\ref{ApK}. \phantom\qedhere
\end{proof}
\begin{figure*}
	\small
	\begin{align}
		\hat{K}_{\text{u-cpt}}\!=\!&\arg\max_{k}\Bigg\{ \frac{1}{4\delta_2^2\!-\!\delta_1^2}\bigg[
		\mathrm{exp}\Big(\frac{\delta_1\big(\alpha_2(\Re\{\mathbf{y}\})\alpha_3(\Re\{\mathbf{y}\})\!-\!\alpha_1(\Re\{\mathbf{y}\})\delta_1\big)\!+\!\delta_2\big(\alpha_2(\Re\{\mathbf{y}\})^2\!+\!\alpha_3(\Re\{\mathbf{y}\})^2\big)\!+\!4\delta_2^2\alpha_1(\Re\{\mathbf{y}\})}{(\delta_1^2-4\delta_2^2)/2}\Big)\nonumber\\
		&+\mathrm{exp}\Big(\frac{\delta_1\big(\alpha_2(\Im\{\mathbf{y}\})\alpha_3(\Im\{\mathbf{y}\})\!-\!\alpha_1(\Im\{\mathbf{y}\})\delta_1\big)\!+\!\delta_2\big(\alpha_2(\Im\{\mathbf{y}\})^2\!+\!\alpha_3(\Im\{\mathbf{y}\})^2\big)\!+\!4\delta_2^2\alpha_1(\Im\{\mathbf{y}\})}{(\delta_1^2-4\delta_2^2)/2}\Big) \bigg]
		\Bigg\} \label{hKF}\\
		\bottomrule\nonumber
	\end{align}
	\vspace{-22mm}
\end{figure*}

Note that the complexity of \eqref{hKF} grows with the number of
coexisting devices since it requires testing all the possible
$Q + 1$ hypothesis.
 But more importantly, \eqref{hKF} is difficult to analyze, thus, deriving performance insights becomes cumbersome. To circumvent these issues, we analyze a much simpler, although sub-optimal, detection approach in the sequel.

The idea is to relax the original detection problem \eqref{hK}, which directly leads to an integer solution, to an estimation problem in the continous real domain.
Since $\mathbb{E}[\mathbf{y}]$ converges to a zero vector, while  $\mathrm{var}[y[n]-w[n]]$ scales linearly with $K$, the (relaxed) minimum-variance unbiased (MVU) estimator of $K$ exploits the sample variance as 
%
%
%
\begin{align}
\widehat{K}_{\mathrm{u-cpt}}^\mathbb{R}&=\big(|\mathbf{s}^H\mathbf{y}|^2/N^2-\sigma^2/N\big)/\bar{p}. \label{UCPTK}
\end{align}
Then, since $K\in\mathbb{N}$, we need to apply a rounding operation, i.e., $\widehat{K}_{\mathrm{u-cpt}}=\mathrm{round}\big(\widehat{K}_{\mathrm{u-cpt}}^\mathbb{R}\big)$, which is in fact the only possible source 
 of performance degradation compared to the optimal NP detector \eqref{hKF}. The reason is simple, \eqref{UCPTK} is not only the MVU but also the ML estimator of $K$, i.e., the value of $K$ that maximizes $f(\mathbf{y};K)$ for $K\in\mathbb{R}$.
In Section~\ref{Dist}, we discuss the specific rounding operation that we adopt in this work and the implications.
\vspace{-4mm}
\subsection{On the Accuracy of the Proposed Estimator}
\vspace{-2mm}
For simplicity, the performance degradation from the rounding operation is neglected in this subsection.
The variance of the estimator can be easily obtained from  \eqref{UCPTK} as
\begin{align}
	\mathrm{var}\big[\widehat{K}^\mathbb{R}_{\mathrm{u-cpt}}\big]&= (K\bar{p}+\sigma^2/N)^2/\bar{p}^2=K^2+\frac{2K}{N\bar{\gamma}}+\frac{1}{N^2\bar{\gamma}^2}\label{var}
\end{align}
by exploiting $\frac{\mathbf{s}^H\mathbf{y}}{N}\sim \mathcal{CN}(0,K\bar{p}+\frac{\sigma^2}{N})$, thus, $\frac{|\mathbf{s}^H\mathbf{y}|^2}{N^2}\sim \mathrm{Exp}\big((K\bar{p}+\frac{\sigma^2}{N})^{-1}\big)$.
%
Since $K\ge 0$,   $\mathrm{var}\big[\widehat{K}_{\mathrm{u-cpt}}\big]$ increases with $K$. In fact, in the high SNR regime and/or for a sufficiently large number of samples $N$, it holds that $\mathrm{var}\big[\widehat{K}^\mathbb{R}_{\mathrm{u-cpt}}\big]\rightarrow K^2$. This means that the $K^2$ contribution to the estimator variance is due to the fading uncertainty, hence,
the uncontrolled fading in the U-CPT scenario is the main cause of estimation inaccuracy as $K$ increases.
This motivates the introduction of more evolved estimators, and corresponding requirements, in the following.
\section{Assisted CPT (A-CPT)}\label{acpt}
\vspace{-2mm}
The main reason behind the limited accuracy of the U-CPT detector/estimator is the channel uncertainty. Herein, we address this by introducing a coordinated training phase common to all devices as illustrated in Fig.~\ref{Fig1}b. Specifically, at the beginning of each time slot, the coordinator sends a broadcast pilot signal $\mathbf{v}\in\mathbb{C}^{N_1}$ comprising $N_1<N$ symbols.
Then, the signal received by the $i-$th device is given by
\begin{align}
z_i[n] = \sqrt{p}h_iv[n] + w_i'[n],\ n=1,2,\cdots,N_1,
\end{align}
where $||\mathbf{v}||^2=N_1$, $p$ is the per-symbol average transmit power of the coordinator, and $w_i'[n]\sim\mathcal{CN}(0,\sigma_i^2)$ is the $n-$th AWGN sample at the $i-$th device. For simplicity, we assume $\sigma_i^2=\sigma^2, \forall i$.
This broadcast pilot transmission phase is leveraged by the active devices to estimate their corresponding channel coefficient, which is assumed UL-DL reciprocal due to a time division duplex operation \cite{Carvalho.2017,Lopez.2020,Chen.2021}. Specifically, the MVU estimator of $h_i$ is given by
\begin{align}
\hat{h}_i=\mathbf{v}^H \mathbf{z}_i/(N_1\sqrt{p})\sim \mathcal{CN}\Big(0,\beta_i+\frac{1}{N_1\varrho}\Big),
\end{align}
where $\varrho\triangleq p/\sigma^2$,
while the estimation error is given by
\begin{align}
\tilde{h}_i=\mathbf{v}^H \mathbf{w}_i'/(N_1\sqrt{p})\sim \mathcal{CN}\Big(0,\frac{1}{N_1\varrho}\Big).
\end{align}

After such initial DL CSI acquisition phase, active devices exploit the remaining $N_2=N-N_1$ symbols for sending the shared pilot sequence $\mathbf{s}\in\mathbb{C}^{N_2}$, with $|s[n]|^2=1\ \forall n$, but phase shifted as $\mathbf{s}_i=e^{-\imath\angle\hat{h}_i}\mathbf{s}=\hat{h}_i^*\mathbf{s}/|\hat{h}_i|$, thus, aiming at a coherent signal combination at the coordinator. 
Meanwhile, the transmit power of the devices is established according to one of the power control mechanisms described in the next subsections. 
%
\vspace{-4mm}
\subsection{Fixed Power Control}
\vspace{-2mm}
Herein, we use the same statistical inverse power control mechanism as in the case of U-CPT, i.e., $p_i=\bar{p}/\beta_i$, thus, the power allocation remains fixed in static deployments. We refer to A-CPT with such fixed power control as A-CPT-F.

Since $\mathbf{s}$ is phase-shifted by the $i-$th (active) device as $\mathbf{s}\hat{h}_i^*/|\hat{h}_i|$, the  signal $\mathbf{y}\in\mathbb{C}^{N_2}$ received at the coordinator  is given by
\begin{align}
y[n] &= \sum_{i\in\mathcal{K}}\sqrt{\bar{p}/\beta_i} \hat{h}_i^*h_i s[n]/|\hat{h}_i|+w[n]\stackrel{(a)}{=}\sum_{i\in\mathcal{K}}\sqrt{\bar{p}/\beta_i}(|\hat{h}_i|-\hat{h}_i^*\tilde{h}_i/|\hat{h}_i|)s[n]+w[n]\nonumber\\
&\stackrel{(b)}{=}\!\sum_{i\in\mathcal{K}}\!\!\sqrt{\!\frac{\bar{p}}{\beta_i}}\big(|\hat{h}_i|\!-\!\tilde{h}_i'\big)s[n]\!+\!w[n],\!\  n\!=\!1,2,\cdots,N_2,\label{yn}
\end{align}
%
where we use 
 $\hat{h}_i=h_i+\tilde{h}_i$
%
and $\tilde{h}_i'\triangleq \frac{\hat{h}_i^*\tilde{h}_i}{|\hat{h}_i|}$ to attain (a) and (b), respectively. 
Note that $\tilde{h}_i'$ is distributed as $\tilde{h}_i$ since $\frac{\hat{h}_i^*}{|\hat{h}_i|}$ is independent of $\tilde{h}_i$ and uniformly distributed in the unit circle. 
%
%

%
\subsubsection{Detector for $K$}\label{kacptf}
Observe that the statistics of $\mathbf{y}$ in \eqref{yn} depend on the specific set of active devices $\mathcal{K}$ rather than on its cardinality alone. Therefore, an optimal detector inevitably requires evaluating the likelihood of all the possible sets of cardinality $k$  for each hypothesis $\mathcal{H}_k$. This is a nondeterministic polynomial-time (np) complete  problem. Since $Q$ can be large, the number of combinations to test when assessing each hypothesis $\mathcal{H}_k$, i.e., $\binom{Q}{k}$, may be huge and computationally unaffordable. 
Therefore, here we do not present the optimum detector and resort directly to a sub-optimal but affordable method.

Similar to our approach in Section~\ref{kucpt}, we relax the detection problem to an estimation problem. For this, let us first define 
\begin{align}
\zeta_\mathcal{K}&\triangleq \mathbf{s}^H\mathbf{y}/N_2=\sum_{i\in\mathcal{K}}\!\!\!\!\!\!\!\underbrace{\sqrt{\bar{p}/\beta_i}|\hat{h}_i|}_{\!\!\!\!\sim \sqrt{\frac{\bar{p}}{2}}\mathrm{Ray}\big(\sqrt{1\!+\!\frac{1}{N_1\bar{\gamma}_i}}\big)}\!\!\!\!\!\!-\sum_{i\in\mathcal{K}}\!\!\!\!\underbrace{\sqrt{\bar{p}/\beta_i}\tilde{h}_i}_{\sim\mathcal{CN}\big(0,\frac{\bar{p}}{N_1\bar{\gamma}_i}\big)}\!+\!\underbrace{\mathbf{s}^H\mathbf{w}/N_2}_{\sim\mathcal{CN}\big(0,\frac{\sigma^2}{N_2}\big)}.\label{zetaK}
\end{align}
where $\bar{\gamma}_i\triangleq \varrho\beta_i$ is the average SNR of the broadcast pilot received at the $i-$th device.
Observe that although $K$ impacts not only the mean but multiple statistics of $\zeta_\mathcal{K}$ via the first two terms of \eqref{zetaK}, we can safely rely only on the mean as a single channel realization/sample is available. 

The first order statistics of $\zeta_\mathcal{K}$ can be computed as follows
\begin{align}
\mathbb{E}[\zeta_\mathcal{K}]&\stackrel{(a)}{=}\mathbb{E}\big[\Re\{\zeta_\mathcal{K}\}\big]=\sum_{i\in\mathcal{K}}\sqrt{\bar{p}/\beta_i}\mathbb{E}\big[|\hat{h}_i|]\stackrel{(b)}{=}\frac{\sqrt{\pi\bar{p}}}{2}\sum_{i\in\mathcal{K}}\sqrt{1+\frac{1}{N_1\bar{\gamma}_i}}\stackrel{(c)}{\approx} \frac{K\sqrt{\pi\bar{p}}}{2}\sqrt{1+\frac{1}{ N_1\bar{\gamma}'}},\label{K1}
\end{align}
where (a) holds since $\Im\{\zeta_\mathcal{K}\}=0$, while (b) follows from $\sqrt{\bar{p}/\beta_i}\mathbb{E}\big[|\hat{h}_i|\big]=\sqrt{\pi\bar{p}(1+1/(N_1\bar{\gamma}_i))}/2$. Thereafter, a relaxing approximation is needed to proceed further since the set of active devices and their corresponding $\{\beta_i\}$ and $\{\bar{\gamma}_i\}$ are unknown. This is done in (c) by exploiting  the relationship between the arithmetic and geometric mean twice, and defining 
\begin{align}
\bar{\gamma}'\triangleq \Big(\frac{1}{Q}\sum_{i\in\mathcal{Q}}\bar{\gamma}_i^{-1}\Big)^{-1}.
\end{align}
Observe that $\bar{\gamma}'$ is the harmonic mean of $\{\bar{\gamma}_i\}_{i\in\mathcal{Q}}$, while  
$1/\bar{\gamma}'$ can be interpreted as a traditional mean-based estimate of $1/\bar{\gamma}_i$. The estimate becomes more accurate as the set $\{\bar{\gamma}_i\}_{i\in\mathcal{Q}}$  becomes less disperse and/or as $\min_{i\in\mathcal{Q}}\bar{\gamma}_i$ increases. The latter is because as $\min_{i\in\mathcal{Q}}\bar{\gamma}_i$ increases, one can more easily assure that terms $\{1+1/(N_1\bar{\gamma}_i)\}$ become more similar among one another (closer to 1). Hence, the arithmetic and geometric means, which were used to go from (b) to (c), converge faster as exemplified in \cite{LopezAlves.2018,Lopez.2019} in the context of an ultra-reliable low-latency scenario. 

Using \eqref{K1} and assuming $\bar{\gamma}'$ is known at (computed beforehand by) the coordinator, the 
A-CPT-F estimator for $K$ 
can be set as
\vspace{-2mm}
\begin{align}
\widehat{K}_{\mathrm{a-cpt-f}}^\mathbb{R}=2\Re\{\zeta_\mathcal{K}\}/\sqrt{\pi\bar{p}\big(1+\frac{1}{N_1\bar{\gamma}'}\big)},\label{ACPTF1}
\end{align}
while $\widehat{K}_{\mathrm{a-cpt-f}}=\mathrm{round}(\widehat{K}_{\mathrm{a-cpt-f}}^\mathbb{R})$.
\vspace{-2mm}
\subsubsection{On the Accuracy of the Proposed Estimator}
Departing from \eqref{zetaK} and \eqref{ACPTF1}, we have that
\begin{align}
\widehat{K}_{\mathrm{a-cpt-f}}^\mathbb{R}\!=\! \upsilon\bigg(\!\sum_{i\in\mathcal{K}}\!\!\!\!\!\!\!\!\!\!\!\!\!\!\underbrace{\frac{|\hat{h}_i|}{\sqrt{\beta_i}}}_{\!\!\!\!\!\frac{1}{\sqrt{2}}\mathrm{Ray}\big(\sqrt{1+\frac{1}{N_1\bar{\gamma}_i}}\big)}\!\!\!\!\!\!\!\!\!\!\!\!\!-\!\sum_{i\in\mathcal{K}}\!\underbrace{\Re\Big\{\!\frac{\tilde{h}_i'}{\sqrt{\beta_i}}\!\Big\}}_{\mathcal{N}\big(0,\frac{1}{2N_1\bar{\gamma}_i}\big)}\!\!\!+\!\underbrace{\Re\Big\{\!\frac{\mathbf{s}^H\mathbf{w}}{N_2\sqrt{\bar{p}}}\!\Big\}}_{\mathcal{N}\big(0,\frac{1}{2N_2\bar{\gamma}_c}\big)}\!\bigg),\label{KacptF}
\end{align}
where $\upsilon\triangleq 2/\sqrt{\pi(1+1/(N_1\bar{\gamma}'))}$, while $\bar{\gamma}_c\triangleq \bar{p}/\sigma^2$ is the target average receive  SNR at the coordinator.
Then,
\vspace{-2mm}
\begin{align}
\mathrm{var}\big[\widehat{K}_{\mathrm{a-cpt-f}}^\mathbb{R}\big] =\upsilon^2\bigg(\!\Big(1\!-\!\frac{\pi}{4}\Big)\!\sum_{i\in\mathcal{K}}\!\Big(1\!+\!\frac{1}{N_1\bar{\gamma}_i}\Big)\!+\!\sum_{i\in\mathcal{K}}\!\frac{1}{2N_1\bar{\gamma}_i}\!+\!\frac{1}{2N_2\bar{\gamma}_c}\!\bigg),\label{varK}
\end{align}
while  \eqref{varK} transforms to
\vspace{-2mm}
\begin{align}
\mathrm{var}\big[\widehat{K}_{\!\mathrm{a-cpt-f}}^\mathbb{R}\big]\!=\! \upsilon^2\bigg(\!\Big(1\!-\!\frac{\pi}{4}\Big)\Big(1\!+\!\frac{1}{N_1\bar{\gamma}'}\Big)\!K\!+\!\frac{1}{2N_1\bar{\gamma}'}\!+\!\frac{1}{2N_2\bar{\gamma}_c}\bigg)=\frac{4-\pi}{\pi}K+\frac{1+\frac{N_1\bar{\gamma}'}{N_2\bar{\gamma}_c}}{\pi(1+N_1\bar{\gamma}')}\label{varKS}
\end{align}
%
for a simplified scenario with $\bar{\gamma}_i=\bar{\gamma}',\ \forall i$, and after performing simple algebraic transformations.

Different from the estimator variance for the case of U-CPT \eqref{var}, which has a quadratic dependence on $K$, here the dependence is linear. Therefore, A-CPT-F estimator provides already some resilience against fading uncertainty. 
\vspace{-6mm}
\subsection{Dynamic Power Control}\label{DPC}
\vspace{-2mm}
Since channel estimates are already available at the devices, a more intuitive approach could rely on setting $p_i=\frac{\rho}{|\hat{h}_i|^2}$ given a fixed $\rho$. However,
%
	the devices may  require to transmit with very high power when channels experience poor conditions, i.e., very small $|\hat{h}_i|$, which may not be affordable.
	Note also that
	$\mathbb{E}[\frac{\rho}{|\hat{h}_i|^2}]$ does not converge, which means 
	that the average power consumption cannot be controlled through $\rho$.
%
To overcome this, we propose silencing the devices experiencing very poor channel conditions. Thus, only those active devices with $|\hat{h}_i|^2\ge \mu_i$ may transmit. This guarantees operating with the same statistics of receive power from each device signal at the coordinator. We refer to A-CPT with such dynamic power control as A-CPT-D. 
\vspace{-3mm}
\begin{theorem}\label{the1}
	The following parameter configuration
	\vspace{-3mm}
	\begin{align}	
	\rho= -\bar{p}(1-\xi)/\mathrm{li}(1-\xi),\ \ \mu_i  = -\vartheta_i\ln(1-\xi),\ \forall i,\label{u1rho}
	\end{align}
	where
	%
	$\vartheta_i\triangleq \beta_i+\frac{1}{N_1\varrho}=\beta_i\big(1+\frac{1}{N_1\bar{\gamma}_i}\big)$,
	%
	guarantees that the devices operate with the same average transmit energy consumption as for U-CPT and A-CPT-F, while constraining the  probability of the active devices of remaining in silence to  $\xi$ as $N_1\bar{\gamma}_i\rightarrow \infty,\ \forall i$.
\end{theorem}
\begin{proof}
See Appendix~\ref{AppT1}. \phantom\qedhere
\end{proof}
\vspace{-2mm}
Some insights on the configuration of $\xi$ are discussed in Section~\ref{scalab}.
\subsubsection{Detector for $K$}
Let us denote by $\mathcal{K}'$ and $\mathcal{K}''$ the set of active devices with $|\hat{h}_i|^2\ge \mu_i$ (UL transmitting devices) and $|\hat{h}_i|^2< \mu_i$ (silent devices), respectively, thus, $\mathcal{K}=\mathcal{K}'\cup \mathcal{K}''$. Then, 
the  signal received at the coordinator, $\mathbf{y}\in\mathbb{C}^{N_2}$, is given by
\vspace{-2mm}
\begin{align}
y[n] = \sum_{i\in\mathcal{K}'}\sqrt{\rho} \hat{h}_i^*h_is[n]/|\hat{h}_i|^2+w[n]=\sqrt{\rho}K's[n]\!+\!\varphi_{\mathcal{K}'}s[n]\!+\!w[n],\ n\!=\!1,2,\cdots,N_2,\label{ypE}
\end{align}
where $K'\triangleq |\mathcal{K}'|$, while last step
follows after using $\hat{h}_i=h_i+\tilde{h}_i$ and setting $\varphi_{\mathcal{K}'}\triangleq \sqrt{\rho}\sum_{i\in\mathcal{K}'}\frac{\hat{h}_i^*\tilde{h}_i}{|\hat{h}_i|^2}$.

Observe that $\mathbf{y}$ in \eqref{ypE} not only depends directly on the cardinality but also on the specific set of transmitting devices $K'$ via $\varphi_{\mathcal{K}'}$. Therefore, our arguments in Section~\ref{kacptf} against the viability of an optimum detector are also valid here, hence, we next resort to a sub-optimal but affordable/efficient detector exploiting a (relaxed) estimation of $K'$. 

Since $\mathbb{E}[\varphi_{\mathcal{K}'}]=0$ and $\mathbb{E}[\Im\{\mathbf{s}^H\mathbf{y}\}]=0$, and discarding the information on $K'$ contained in $\varphi_{\mathcal{K}'}$, we can directly estimate the relaxed $K'$
as
\vspace{-2mm}
\begin{align}
\widehat{K'}_{\mathrm{a-cpt-d}}^\mathbb{R}=\Re\{\mathbf{s}^H\mathbf{y}\}/(N_2\sqrt{\rho})\label{KpE}
\end{align}
while, $\widehat{K'}_{\mathrm{a-cpt-d}}=\mathrm{round}(\widehat{K'}_{\mathrm{a-cpt-d}}^\mathbb{R})$.
To proceed further, we assume all devices are configured such that $\xi_i=\xi,\forall i$ (i.e., by adopting the parameter configuration given in Theorem~\ref{the1}). Then, 
the following result establishes a relatively simple estimator for $K''$. Nevertheless, readers may refer to Appendix~\ref{AppT2}, not only to follow Theorem~\ref{th2}'s proof, but also to get insights on how the estimator of $K''$ could be computed for a general case where $\xi_i\ne\xi_j$, for some $i,j\in\mathcal{K}$.  
\vspace{-2mm}
\begin{theorem}\label{th2}
	Assume $\xi_i=\xi,\forall i$, then the MMSE estimate of $K''$ is given by
	\vspace{-2mm}
	\begin{align}
	\widehat{K''}_\mathrm{\!\!a-cpt-d}^\mathbb{R}&\!=\frac{\xi}{(1\!-\!\xi)^2}\Bigg(\xi^{K_l}\bigg((\xi\!-\!1)(K_l\!+\!1)\binom{1\!+\!K_l\!+\!\widehat{K'}}{1+K_l}\ _2F_1\big(1\!+\!K_l,-\widehat{K'},2\!+\!K_l,\xi\big)\!-\!\binom{2\!+\!K_l\!+\!\widehat{K'}}{2+K_l}\nonumber\\
	\vspace{-4mm}
	&\qquad\qquad\qquad\qquad\qquad\qquad\times\xi\!\ _2F_1\big(1\!+\!K_l,-\widehat{K'},3\!+\!K_l,\xi\big)\!\bigg)\!\!+\!\widehat{K'}\!+\!1 \!\Bigg),\label{mmse}
	\end{align}
	where $K_l$ is the solution of
	\vspace{-3mm}
	\begin{align}
	B_{\xi}\big(1+K_l,1+\widehat{K'}\big)(1+K_l+\widehat{K'})!=\xi K_l!\widehat{K'}!.\label{Kl}
	\end{align}	
\end{theorem}
\vspace{-6mm}
\begin{proof}
	See Appendix~\ref{AppT2}. \phantom\qedhere
\end{proof}
\vspace{-2mm}
\vspace{-4mm}
\begin{corollary}
	For a sufficiently large $K_l$, as expected for small $\xi$, such result simplifies to
	\vspace{-2mm}
	\begin{align}
	\widehat{K''}_\mathrm{a-cpt-d}^\mathbb{R}\stackrel{K_l \rightarrow \infty}{=} \xi\big(1+\widehat{K'}_\mathrm{a-cpt-d}^\mathbb{R}\big)/(1-\xi)^2\label{Kp2}
	\end{align}
\end{corollary}
\begin{proof}
	Rather than taking $K_l\rightarrow\infty$ in \eqref{mmse}, use \eqref{mmse0} in Appendix~\ref{AppT2}.
\end{proof}
%
%

Since \eqref{Kp2} holds accurately already for small $\xi$, hereinafter we adopt it to simplify our exposition and analysis. 
Note that we are using $\widehat{K'}_\mathrm{a-cpt-d}^\mathbb{R}$ instead of $\widehat{K'}_\mathrm{a-cpt-d}$ aiming to mitigate the accumulation of rounding errors. 
From \eqref{Kp2}, one can notice that as $\widehat{K'}_\mathrm{a-cpt-d}$ increases, $\mathbb{P}(K\ne K')$ increases as well, motivating the introduction of the  correcting term $\widehat{K''}$.
Finally, 
\begin{align}
\widehat{K}_\mathrm{a-cpt-d}^\mathbb{R}&=\widehat{K'}_\mathrm{a-cpt-d}^\mathbb{R}+\widehat{K''}_\mathrm{a-cpt-d}^\mathbb{R},\label{varKT}
\end{align}
and $\widehat{K}_\mathrm{a-cpt-d}=\mathrm{round}(\widehat{K}_\mathrm{a-cpt-d}^\mathbb{R})$.
%
%
\subsubsection{On the Accuracy of the Proposed Estimator}
\vspace{-1mm}
%
\begin{theorem}\label{TH3}
	\vspace{-4mm}
	The variance of the A-CPT-D estimator  given in \eqref{varKT}
	is approximately given by
	\begin{align}
	\mathrm{var}[\widehat{K}_{\mathrm{a-cpt-d}}^\mathbb{R}]\!=\!\frac{\psi}{2}\Big(\frac{1}{N_2\bar{\gamma}_c'}\!-\!\frac{\mathrm{li}(1\!-\!\xi)}{1-\xi}\sum_{i\in\mathcal{K}'}\frac{1}{1\!+\!N_1\bar{\gamma}_i}\Big)\label{varF}
	\end{align}
	for relatively small $\xi$, where $\psi\triangleq 1+\xi^2/(1-\xi)^4$ and $\bar{\gamma}_c'\triangleq \rho/\sigma^2$ is the target average receive  SNR at the coordinator.
\end{theorem}
\vspace{-4mm}
\begin{proof}
	See Appendix~\ref{AppTH3}.
\end{proof}

Note that he variance of the A-CPT-D estimator increases linearly with $K'$ due to the sum operator in \eqref{varF}. In fact, for the specific scenario where devices' deployment is homogeneous in the sense that $\bar{\gamma}_i=\bar{\gamma}'$,
\eqref{varF} simplifies to
\begin{align}
\mathrm{var}[\widehat{K}_{\mathrm{a-cpt-d}}^\mathbb{R}]&=\frac{\psi}{2}\Big(\frac{1}{N_2\bar{\gamma}_c'}\!-\!\frac{\mathrm{li}(1-\xi)K'}{(1+N_1\bar{\gamma}')(1-\xi)}\Big)\label{varF2}.
\end{align}
Finally, one can get more insights on the variance performance specified by \eqref{varF2} by leveraging \cite[eq. (6.8.1)]{Olver.2010} to attain
\begin{align}
\mathrm{var}[\widehat{K}_{\mathrm{a-cpt-d}}^\mathbb{R}]&<\frac{\psi/2}{1+ N_1\bar{\gamma}'}\ln\Big(1-\frac{1}{\ln(1-\xi)}\Big)K'+\frac{\psi}{2N_2\bar{\gamma}_c'}\approx\frac{\psi \ln(1+\xi)K'}{2(1+ N_1\bar{\gamma}')}+\frac{\psi}{2N_2\bar{\gamma}_c'},\label{vb}
\end{align}
where last step comes from leveraging $-\ln(1-\xi)\approx \xi$ which holds accurate for $\xi\ll 1$.
Since $N_1\bar{\gamma}'\gg 1$ holds in practice, \eqref{vb} leads to significantly small variance bound values.
\vspace{-5mm}
\section{On the Distribution of the Proposed Detectors}\label{Dist}
\vspace{-2mm}
In general, and especially for the considered setup, the estimator variance cannot be directly used to quantify, at least thoroughly, the performance degradation due to detection mismatches. Instead, the distribution of the classification results must be leveraged. 

The PMF of $\widehat{K}$ can be found as
\vspace{-2mm}
\begin{align}
	p_{\widehat{K}}(\hat{k})&=\mathbb{P}\big(\mathrm{round}(\widehat{K}^\mathbb{R})=\hat{k}\big),
\end{align}
where $\mathrm{round}(\cdot)$ denotes a rounding rule. Herein, we adopt the following rounding rules:
\begin{itemize}
	\item \textit{Rounding to nearest integer (NI)}, where $\mathrm{round}(x)$ equals $\lfloor x \rfloor$ or $\lceil x \rceil$ if the fractional part of $x$ is respectively smaller or greater than $1/2$.
	\vspace{-2mm}
Then,
\vspace{-2mm}
\begin{align}
	p_{\widehat{K}}(\hat{k})&=\mathbb{P}\big(\hat{k}-1/2\le \widehat{K}^\mathbb{R}\le \hat{k}+1/2\big)=F_{\widehat{K}^\mathbb{R}}(\hat{k}+1/2)-F_{\widehat{K}^\mathbb{R}}(\hat{k}-1/2).
\end{align}
\vspace{-12mm}
\item \textit{Maximum likelihood (ML)-based rounding}, where 
	\vspace{-3mm}
\begin{align}
	\widehat{K}=\arg \max_{k\in\mathbb{N}}f_{\widehat{K}^\mathbb{R}|K=k}(\hat{k}).\label{mlr}
\end{align}
Here, $f_{\widehat{K}^\mathbb{R}|K=k}(\hat{k})$ is the likelihood function of $\widehat{K}^\mathbb{R}$ given the outcome $k$ of $K$.
\vspace{-2mm}
%
%
\end{itemize}

Observe that the distribution of the relaxed estimator, $\hat{K}^\mathbb{R}$, is needed for computing $p_{\widehat{K}}(\hat{k})$ under both rounding approaches. Thus, our efforts are in this direction in the following subsections. 
Interestingly, different from the NI-based rounding, 
the distribution under ML-based rounding impacts the rounding operation itself. Hence, it becomes cumbersome finding $p_{\widehat{K}}(\hat{k})$ for the latter.
Moreover, note that due to the intrinsic similarities between \eqref{mlr} and \eqref{hK}, the relaxed estimators followed by ML-based rounding are expected to perform similar to the optimal NP detectors. This is verified in the case of U-CPT, for which the optimal detector was derived in Theorem~\ref{EK}, in Section~\ref{results}, specifically Table~\ref{table3}.
	\vspace{-5mm}
\subsection{U-CPT}
	\vspace{-2mm}
From \eqref{UCPTK}, it follows directly that $\widehat{K}_{\mathrm{u-cpt}}^\mathbb{R}$ is a shifted exponential RV since $|\mathbf{s}^H\mathbf{y}|^2/N^2\sim \mathrm{Exp}(1/(K\bar{p}+\sigma^2))$. Specifically,
\begin{align}
F_{\widehat{K}^\mathbb{R}_{\mathrm{u-cpt}}|K}(\hat{k})=1-e^{-\frac{1+N\bar{\gamma}\hat{k}}{1+N\bar{\gamma}K}},\ \ 
f_{\widehat{K}^\mathbb{R}_{\mathrm{u-cpt}}|K}(\hat{k})=\frac{N\bar{\gamma}}{1+N\bar{\gamma}K}e^{-\frac{1+N\bar{\gamma}\hat{k}}{1+N\bar{\gamma}K}},\ \ \hat{k}\ge -1/(N\bar{\gamma}).
\end{align}
%
	\vspace{-14mm}
\subsection{A-CPT-F}\label{acptf}
	\vspace{-2mm}
According to \eqref{KacptF}, $\widehat{K}_\mathrm{a-cpt-f}^\mathbb{R}$ can be written as 
\vspace{-2mm}
\begin{align}
\sqrt{2}\widehat{K}_\mathrm{a-cpt-f}^\mathbb{R}/\upsilon = \sum\nolimits_{i\in\mathcal{K}}U_i+R,\label{KF}
\end{align}
where 
\vspace{-4mm}
\begin{align}
U_i\triangleq |\hat{h}_i/\sqrt{\beta_i}|\sim \mathrm{Ray}(\sqrt{1+1/(N_1\bar{\gamma}_i)}),\ \ 
R\triangleq \frac{\Re\{\mathbf{s}^H\mathbf{w}\}}{N_2\sqrt{\bar{p}}}-\sum_{i\in\mathcal{K}}\frac{\Re\{\tilde{h}_i\}}{\sqrt{\beta_i}}\sim \mathcal{N}(0,\phi_\mathcal{K}^2),\label{KF2}
\end{align}
with $\phi_\mathcal{K}\triangleq \big(\frac{1}{N_2\bar{\gamma}_c}+\frac{1}{N_1}\sum_{i\in\mathcal{K}}\frac{1}{\bar{\gamma}_i}\big)^{1/2}$.
\vspace{-2mm}
\begin{theorem}\label{a-cpt-f-T}
	The distribution of $\widehat{K}_\mathrm{a-cpt-f}^\mathbb{R}$ conditioned on $K$ is given by
	\begin{align}
	F_{\widehat{K}_{\mathrm{a-cpt-f}}^\mathbb{R}|K}\!(\hat{k})
	\!=\!1\!-\!\Biggl[\sum_{\substack{\forall \mathcal{K}\subseteq \mathcal{Q}\\|\mathcal{K}|=K}}\underbrace{\!\int\limits_{0}^{\infty}\!\cdots\!\int\limits_0^\infty}_{K\ \text{integrals}}\!\!\mathrm{erfc}\Big(\frac{\hat{k}}{\phi_\mathcal{K}\upsilon}\!-\!\sum_{i\in\mathcal{K}}\!\frac{u_i}{\sqrt{2}\phi_\mathcal{K}}\Big) e^{-\sum\limits_{i\in\mathcal{K}}\frac{u_i^2/2}{1+1/(N_1\bar{\gamma}_i)}}\prod_{i\in\mathcal{K}}\!u_i \mathrm{d}u_i \Biggl]\frac{\varsigma_\mathcal{K}}{2\binom{Q}{K}},\label{final}
	\end{align}
	where $\varsigma_\mathcal{K}\triangleq \prod_{i\in\mathcal{K}} (1+1/(N_1\bar{\gamma}_i))^{-1}$. 
\end{theorem}
\vspace{-6mm}
\begin{proof}
	See Appendix~\ref{TH2}. \phantom\qedhere
\end{proof}
\vspace{-2mm}

The PDF of $\widehat{K}_{\mathrm{a-cpt-f}}^\mathbb{R}|K$ can be obtained from \eqref{final} as $f_{\widehat{K}_{\mathrm{a-cpt-f}}^\mathbb{R}|K}(\hat{k})=d F_{\widehat{K}_{\mathrm{a-cpt-f}}^\mathbb{R}|K}(\hat{k})/d x$. 
Unfortunately, evaluating \eqref{final} becomes computationally expensive and often unaffordable, specially when i) $Q\gg 1$ since the number of combinations grows exponentially, and/or ii) $K\gg 1$ because of the increased number of integration operations in \eqref{FKYXi}. Therefore, resorting to approximation approaches is necessary for practical algorithms.

\begin{figure}[t!]
	\centering
	\includegraphics[width=0.45\textwidth]{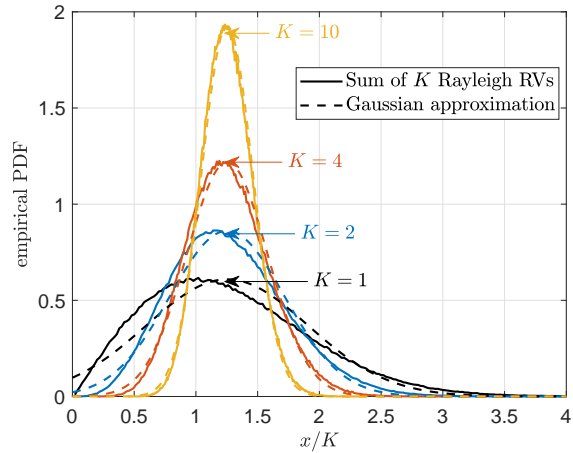}
	\vspace{-4mm}
	\caption{Empirical PDF of the normalized sum of $K$ Rayleigh RVs with $s=1$, and corresponding Gaussian distribution fitting.}
	\label{Fig2}
	\vspace{-8mm}
\end{figure}
A simple approach relies on \textit{homogenizing} the setup by  adopting $\bar{\gamma}'$ as an approximation for $\bar{\gamma}_i$ $\forall i$. Note that the approximation holds with equality when $\beta_i=\beta,\ \forall i$. Under the above assumption, $\{U_i\}$ are i.i.d. such that $U_i\sim \mathrm{Ray}\big(\sqrt{1+1/(N_1\bar{\gamma}')}\big)$, while $R\sim \mathcal{N}\big(0,1/(N_2\bar{\gamma}_c)+K/(N_1\bar{\gamma}')\big)$. Even in highly heterogeneous setups, where $\bar{\gamma}_i$ oscillates significantly for different devices $i$, this approach remains accurate in high SNR regimes where $\min_{i\in\mathcal{Q}}\bar{\gamma}_i\gg 1$ since $1/(N_1\bar{\gamma}_i)\rightarrow 0$ and $1/(N_1\bar{\gamma}')\rightarrow 0$.
%
%

The above approach solves the combinatorial problem since the distribution becomes dependent on the number, rather than the specific set, of active devices.
%
To avoid the multiple integration required for computing \eqref{final},
we resort to fit the Rayleigh distribution into a Gaussian one, such that
%
$\mathrm{Ray}(s)  \rightarrow \mathcal{N}\big(\sqrt{\pi/2}s,(4-\pi)s^2/2\big)$.
Such fitting, which is illustrated in Fig.~\ref{Fig2}, is motivated by: i) a Rayleigh PDF, although not strictly symmetrical, is bell-shaped, and ii) the sum of arbitrarily distributed RVs converge asymptotically (as $K$ increases) to a Gaussian distribution. The accuracy of our approach is assessed later in Section~\ref{results}.

Finally, $\widehat{K}_\mathrm{a-cpt-f}^\mathbb{R}$ is approximately distributed as
\vspace{-2mm}
\begin{align}
\widehat{K}_\mathrm{a-cpt-f}^\mathbb{R}\!&\sim\!  \mathcal{N}\Big(K,\mathrm{var}\big[\widehat{K}_\mathrm{a-cpt-f}^\mathbb{R}\big]\Big),\label{N}
\end{align}
where $\mathrm{var}\big[\widehat{K}_\mathrm{a-cpt-f}^\mathbb{R}\big]$ is given in 
\eqref{varKS}.
%
\vspace{-4mm}
\subsection{A-CPT-D}
\vspace{-2mm}
Observe that
\vspace{-4mm}
\begin{align}
	F_{\widehat{K}_\mathrm{a-cpt-d}^\mathbb{R}|K}(\hat{k})&=\sum_{k'=0}^{K}F_{\widehat{K}_\mathrm{a-cpt-d}^\mathbb{R}|K'}(\hat{k})\ \!p_{K'|K}(k'),
\end{align}
where $p_{K'|K}(k')=\binom{K}{k'}(1-\xi)^{k'}\xi^{K-k'}$. Meanwhile, by leveraging \eqref{Kp2} and \eqref{varKT}, we have that
\begin{align}
	F_{\widehat{K}_\mathrm{a-cpt-d}^\mathbb{R}|K'}(\hat{k})
	&=\mathbb{P}\big(\widehat{K'}_\mathrm{a-cpt-d}^\mathbb{R}+\widehat{K''}_\mathrm{a-cpt-d}^\mathbb{R}\le \hat{k}\ \!\big|\ \!K'\big)\approx\!\mathbb{P}\Big(\widehat{K'}_\mathrm{\!\!a-cpt-d}^\mathbb{R}\!+\!\frac{\xi(1\!+\!\widehat{K'}_\mathrm{\!\!a-cpt-d}^\mathbb{R})}{(1-\xi)^2}\!\le\! \hat{k}\ \!\big|\ \!K'\Big)\nonumber\\
	&=F_{\widehat{K'}_\mathrm{a-cpt-d}^\mathbb{R}|K'}\Big(\frac{\hat{k}(1-\xi)^2-\xi}{(1-\xi)^2+\xi}\Big).\label{FK}
\end{align}
%
%

The problem translates to find $F_{\widehat{K'}_\mathrm{a-cpt-d}^\mathbb{R}|K'}(\hat{k})$. For this, let us denote $X_i\triangleq \Re\{\tilde{h}_i\}/\sqrt{2N_1\varrho}$, $Y_i\triangleq |\hat{h}_i|/\sqrt{\vartheta_i}\ \big|\ |\hat{h}_i|^2\ge \mu_i/\vartheta_i$, while $V\triangleq \Re\{\mathbf{s}^H\mathbf{w}/(N_2\sqrt{\rho})\}$. Then,

%

%
\begin{align}
\widehat{K'}_\mathrm{a-cpt-d}^\mathbb{R}&=K'+\sum_{i\in\mathcal{K}'}\frac{X_i/Y_i}{\sqrt{2N_1\varrho\vartheta_i}}+V=K'+\sum_{i\in\mathcal{K}'}\frac{Z_i}{\sqrt{2N_1\varrho\vartheta_i}}+V,\label{KXY0}
\end{align}
where $Z_i\triangleq X_i/Y_i$. Herein, $X_i\sim \mathcal{N}(0,1)$, $V\sim\mathcal{N}(0,1/(2N_2\bar{\gamma}_c'))$, while $f_{Y_i}(y)=f_{|\hat{h}_i|}(y)/\big(1-F_{|\hat{h}_i|}(\sqrt{\mu_i/\vartheta_i})\big)$ for $y\ge \sqrt{\mu_i/\vartheta_i}$. Hence,
%
%
\vspace{-2mm}
\begin{align}
f_{Z_i}(z)&=\frac{d}{dz}F_{Z_i}(z)=\frac{d}{dz}\mathbb{P}(X_i/Y_i\le z)=\frac{d}{dz}\int_{\sqrt{\mu_i/\vartheta_i}}^{\infty}F_{X_i}(yz)f_{Y_i}(y)\mathrm{d}y\stackrel{(a)}{=}\int_{\sqrt{\mu_i/\vartheta_i}}^{\infty}yf_{X_i}(yz)f_{Y_i}(y)\mathrm{d}y\nonumber\\
&\stackrel{(b)}{=}\sqrt{\frac{2}{\pi}}e^{\mu_i/\vartheta_i}\int_{\sqrt{\mu_i/\vartheta_i}}^{\infty}y^2e^{-y^2(z^2/2+1)}\mathrm{d}y\stackrel{(c)}{=}\frac{e^{\mu_i/\vartheta_i}}{\sqrt{2\pi}}\Big(\frac{z^2}{2}+1\Big)^{-3/2}\Gamma\bigg(\frac{3}{2},\Big(\frac{z^2}{2}+1\Big)y^2\bigg)\bigg|_{y=\sqrt{\mu_i/\vartheta_i}}^{y\rightarrow\infty}\nonumber\\
&\stackrel{(d)}{=}e^{\mu_i/\vartheta_i}\Gamma\big(3/2,(z^2/2+1)\mu_i/\vartheta_i\big)/\sqrt{2\pi(z^2/2+1)^3},
\label{fZ2}
\end{align}
where (a) comes from differentiating under the integral sign by leveraging Leibniz rule and $dF_{X_i}(yz)/dz=yf_{X_i}(yz)$. Then, (b) follows from substituting $f_{X_i}(x)=e^{-x^2/2}/\sqrt{2\pi}$, and $f_{Y_i}(y)$ as given in 
\eqref{fY}. The indefinite integral is solved in (c) using \cite[eq. (2.325.6)]{Gradshteyn.2014}, while (d) follows directly after evaluating the definite integral limits. 
Observe that by assuming the parameter configuration given in Theorem~\ref{the1}, \eqref{fZ2} transforms to
\begin{align}
	f_{Z_i}(z)=\frac{\Gamma\big(3/2,-(z^2/2+1)\ln(1-\xi)\big)}{(1-\xi)\sqrt{2\pi(z^2/2+1)^3}},\ \forall i.\label{fZ22}
\end{align}
\vspace{-12mm}
\begin{theorem}\label{a-cpt-d-T}
	The distribution of $\widehat{K'}_\mathrm{a-cpt-d}^\mathbb{R}$ conditioned on $K'$ is approximately given by
	\vspace{-3mm}
	\begin{align}
F_{\widehat{K'}_{\!\!\mathrm{a-cpt-d}}^\mathbb{R}|K'}(\hat{k'})\!&\approx 1\!-\!\frac{1}{2\binom{Q}{K'}}\!\sum_{\substack{\forall \mathcal{K}'\subseteq \mathcal{Q}\\|\mathcal{K}'|=K'}}\underbrace{\int\limits_{-\infty}^\infty\!\!\!\cdots\!\!\!\int\limits_{-\infty}^\infty}_{K'\ \text{integrals}}\!\!\! \mathrm{erfc}\bigg(\!\sqrt{N_2\bar{\gamma}_c'} \Big(\hat{k'}\!-\!K'\!-\!\sum_{i\in\mathcal{K}'}\!z_i\Big)\bigg)\!\prod_{i\in\mathcal{K}'}\!f_{Z_i}(z_i)\mathrm{d}z_i\label{KD}
\end{align}
for relatively small $\xi$.
\end{theorem}

\begin{proof}
	See Appendix~\ref{TH5}. \phantom\qedhere
\end{proof}
Note that \eqref{KD} is often too computationally expensive, similar to \eqref{final} in Subsection~\ref{acptf} (see related discussions therein), thus, we resort to additional approximations in the following. 
%
%
\vspace{-2mm}
\begin{theorem}\label{TH6}
	Let $T\triangleq \sum_{i\in\mathcal{K}'}Z_i$, and assume  $Z_i\sim Z$ with PDF given in \eqref{fZ22}, where $\xi_i=\xi$ $\forall i$. 
	Then, $T$ is approximately distributed as a scaled Student's $t$ distribution $\sqrt{w_1\Big(1-\frac{2}{\nu}\Big)} \mathcal{T}(\nu)$, where $\nu$ is the solution to
	\vspace{-2mm}
	\begin{align}
	2^{\frac{\nu}{2}-1}\omega_2\Gamma\Big(\frac{\nu}{2}\Big)&=K_{\frac{\nu}{2}}(26\sqrt{\omega_1(\nu-2)})(26\sqrt{\omega_1(\nu-2)})^{\frac{\nu}{2}},\label{eqV}
	\end{align}
	and 
	\vspace{-4mm}
	\begin{align}
	\omega_1\triangleq \mathrm{li}(1-\xi)K'/(1-\xi),\ \  \ \ 
	\omega_2\triangleq \bigg[2\int_{0}^{\infty}\cos (26z)f_Z(z)\mathrm{d}z\bigg]^{K'}. \label{omega1omega2}
	\end{align}
\end{theorem}
\begin{proof}
	See Appendix~\ref{AL3} for the proof and accuracy-related discussions.  \phantom\qedhere
\end{proof}
 
Let us denote $T'\sim \mathcal{T}(\nu)$, where $\nu$ is the solution to \eqref{eqV}. Then, by exploiting Theorem~\ref{TH6}, one attains
 \begin{align}
 	F_{\widehat{K'}_{\!\!\mathrm{a-cpt-d}}^\mathbb{R}|K'}(\hat{k'})&\!\stackrel{(a)}{\approx}
 	\mathbb{P}\bigg(V\le \hat{k'}\!-\!K'\!-\!\sqrt{\frac{\omega_1(1-2/\nu)}{2N_1\varrho\vartheta}}T'\bigg)\!=\!\!\int\limits_{-\infty}^{\infty}\!\!F_V\bigg(\hat{k'}\!-\!K'\!-\!\sqrt{\frac{\omega_1\big(1\!-\!\frac{2}{\nu}\big)}{2N_1\varrho\vartheta}}t\bigg)f_T(t)\mathrm{d}t\nonumber\\
 	&\stackrel{(b)}{=}1\!-\!\frac{\Gamma\big(\frac{\nu+1}{2}\big)}{2\sqrt{\nu\pi}\Gamma(\frac{\nu}{2})}\!\!\int\limits_{-\infty}^{\infty}\!\!\!\mathrm{erfc}\Bigg(\!\!\sqrt{N_2\bar{\gamma}_c'}\bigg(\hat{k'}\!\!-\!K \!  -\!\sqrt{\frac{\omega_1\big(1\!-\!\frac{2}{\nu}\big)}{2N_1\varrho\vartheta}}t\bigg)\!\Bigg)\Big(1\!+\!\frac{t^2}{\nu}\Big)^{\!\!-\frac{\nu\!+\!1}{2}}\!\mathrm{d}t,\label{intK}
 \end{align}
where (a) leverages \eqref{KXY0}, while (b) follows from using the distribution of $V$ and $T'$. By leveraging  \cite[eq. (3)]{Forchini.2008}, one can state the integral operation in \eqref{intK} as an infinite sum that includes factorials, incomplete gamma, and confluent hypergeometric functions. However, such an approach may not significantly reduce the mathematical complexity from numerically computing \eqref{intK}, thus, we do not adopt it here.

Observe that computing \eqref{intK} is significantly much less computationally demanding than \eqref{KD} since i) there is no combinatorial sub-set selection (as that required for selecting the adding terms in \eqref{KD}), and ii) the number of integral computations is reduced to two, i.e., \eqref{intK} and $\omega_2$ in \eqref{omega1omega2}, independently of the value of $K'$.
\vspace{-5mm}
\section{Numerical Results}\label{results}
\vspace{-2mm}
Herein, we analyze numerically the performance of the discussed CPT mechanisms, mostly in terms of the detection success probability given by $p_{\widehat{K}}(K)$. 
Unless stated otherwise, we adopt an NI-based rounding and consider a massive deployment of $Q=1000$ devices, out of which $K=5$ become active at each time slot. 
The devices are uniformly randomly distributed in an annulus region around the coordinator with 25 m and 500 m of inner and outer radius, respectively, as in \cite{Kamil.2018}. The average power gain of the $i-$th channel is modeled as 
%
	$\beta_i=130 + 37.6\times \log_{10}(d_i)\ [\text{dB}]$,
%
where $d_i$ is the distance between the $i-$th device and the coordinator \cite{Kamil.2018}.
%
Moreover, we set $p=30$ dBm, $N=6$ ($N_1=2$, $N_2=4$), $\bar{p}=-110$ dBm, $\xi=10^{-2}$, and $\sigma^2=-120$ dBm by assuming a transmission bandwidth of 200~kHz. Finally, we adopt the parameter configuration given in Theorem~\ref{the1} for a fair performance comparison of the CPT mechanisms. 
%

The results to be discussed next reveal that even when the large-scale fading is heterogeneous across the network, the detectors perform as expected, and the accuracy of the derived theoretical results is not severely affected.
%
\vspace{-6mm}
\subsection{On the PMF of $\widehat{K}$ and Rounding Performance Impact}
\vspace{-2mm}
\begin{figure}[t!]
	\centering
	\includegraphics[width=0.49\textwidth]{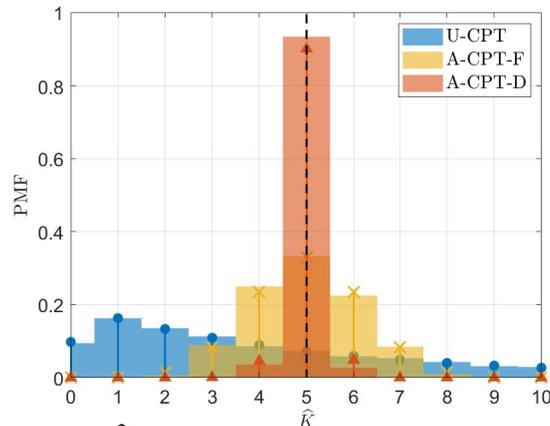}
	\vspace{-6mm}
	\caption{PMF of $\widehat{K}$. 
		 Markers denote Monte Carlo simulations.}
	\label{FigR1}
		\vspace{-6mm}
\end{figure}
Fig.~\ref{FigR1} depicts the PMF of $\widehat{K}$.
Observe that the  derived theoretical results agree with  Monte Carlo simulations, specially in the case of U-CPT, where no approximations were exploited for deriving $p_{\widehat{K}}(\hat{k})$. Results here corroborate the previous observations in Section~\ref{Dist} regarding the distribution of $\widehat{K}^\mathbb{R}$ (and approximately that of $\widehat{K}$): i) its exponential-like shape in the case of U-CPT, and ii) its symmetric-like shape in the case of A-CPT-F (Gaussian) and A-CPT-D (Student's $t$). Note that A-CPT mechanisms are the most appealing, specially A-CPT-D, since $p_{\widehat{K}}(\hat{k})$ is the narrowest around $\widehat{K}=K=5$. This evinces the need of two-stage detection mechanisms, consisting of DL training and UL CPT as depicted in Fig.~\ref{Fig1}b, to mitigate fading uncertainty, thus, enabling more successful detections. The superiority of A-CPT-D is also illustrated in Table~\ref{table3} and in most of the remaining figures. Specifically, Table~\ref{table3} illustrates the attainable performance under other rounding approaches in terms of $p_{\widehat{K}}(K)$. Interestingly, the optimal U-CPT detector specified in Theorem~\ref{EK} performs exactly the same as the simpler relaxed estimator followed by ML rounding. In all the cases, the simple NI-based rounding already provides a performance similar to that offered by the more sophisticated  ML-based rounding. 
%
%
\begin{table}[!t]
	\centering
	\caption{Detection success probability based on relaxed estimation followed by different rounding mechanisms}
	\begin{tabular}{p{1.4cm}|p{1.8cm}p{1.8cm}|p{1.2cm}}
		\toprule
		& \multicolumn{2}{c|}{Relaxed Estimator $+$} & 	  \\ 
		Mechanism & NI rounding & ML rounding & Optimum \\
		\midrule	
		U-CPT & $7.17\%$ & $7.24\%$ & $7.24\%$ \\
		A-CPT-F & $31.70\%$ & $32.50\%$ & --    \\
		A-CPT-D & $91.69\%$ & $92.70\%$ & --   \\
		\bottomrule
	\end{tabular}\label{table3}
\end{table}
\vspace{-6mm}
\subsection{On the Detection Scalability}\label{scalab}
\vspace{-2mm}
\begin{figure}[t!]
	\centering
	\includegraphics[width=0.45\textwidth]{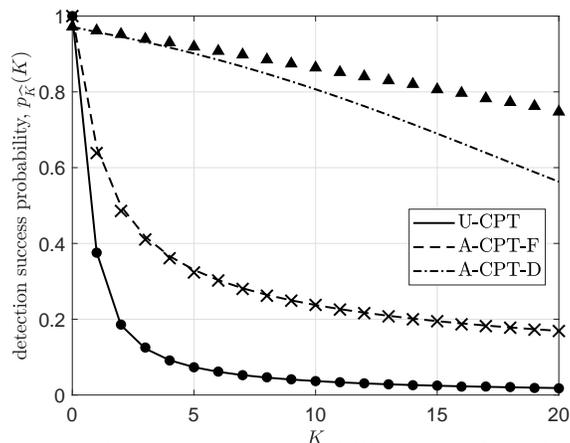}
		\vspace{-6mm}
	\caption{Detection success probability as a function of $K$. 
	}
	\label{FigR2}
		\vspace{-6mm}
\end{figure}
Fig.~\ref{FigR2} illustrates the detrimental effect that the number of device activations has on the detection success probability. This comes as a direct consequence of the estimation variance increasing with $K$ in \eqref{var}, \eqref{varKS}, and \eqref{varF2} for U-CPT, A-CPT-F, and A-CPT-D, respectively. The performance deterioration with $K$ is more abrupt under U-CPT for which the estimation variance increases with $K^2$, while A-CPT mechanisms offer more resilience since the estimation variance increases linearly with $K$. 
In the case of A-CPT-D, the estimation variance increases not only linearly with $K$ but also with the smallest slope,
which finally translates to the best performance. 


\begin{figure}[t!]
	\centering
	\includegraphics[width=0.45\textwidth]{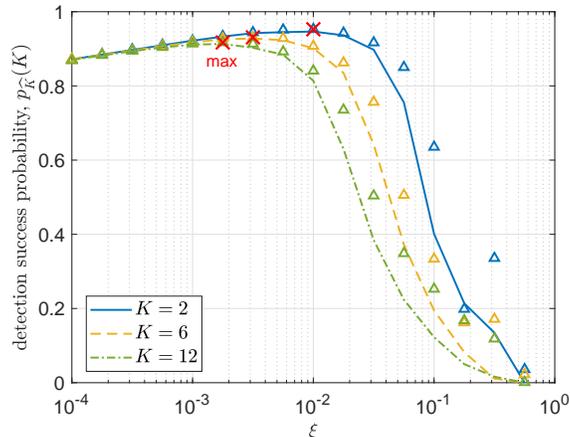}
		\vspace{-6mm}
	\caption{Detection success probability under A-CPT-D as a function of the transmission probability of active devices for $K\in\{2,6,12\}$.}
	\label{FigR5}
		\vspace{-6mm}
\end{figure}
Fig.~\ref{FigR5} shows the impact of $\xi$ on the detection performance. Observe that as $\xi$ decreases approaching 0, $\mu$ decreases, thus, allowing almost every device to transmit even when experiencing poor channel conditions. This means that to satisfy the average transmit power constraint $\rho \mathbb{E}[|\hat{h}_i|^{-2}\ \!|\!\ |\hat{h}_i|^2\ge \mu]=1$, one must accordingly decrease $\rho$, which in turns affects the receive SNR at the coordinator and the detection performance under A-CPT-D. On the other hand, as $\xi$ increases approaching 1, $\mu$ increases, while increasingly restricting the number $K'$ of active devices  that are transmitting. As a result, the detection of $K''$, and ultimately that of $K$, becomes more seriously affected. Based on previous arguments, there is obviously at least one value of $\xi$ that guarantees optimum detection performance. In fact, that point is unique as illustrated in Fig.~\ref{FigR5}, and the larger $K$ is, the smaller the optimum $\xi$ becomes. 
In practice, such value can be configured such that the detection performance for a particular range of values of $K$ remains favorable. For instance, based on the results illustrated in Fig.~\ref{FigR5}, we can assure that $\xi\approx 3\times 10^{-3}$ guarantees a good detection performance $\forall K\le 12$ (presumably also for greater $K$). If there are some statistical priors on $K$, then such information can be also exploited for optimizing $\xi$.
%
\vspace{-4mm}
\subsection{How Many CPT Symbols are Needed?}
\vspace{-2mm}
\begin{figure}[t!]
	\centering
	\includegraphics[width=0.45\textwidth]{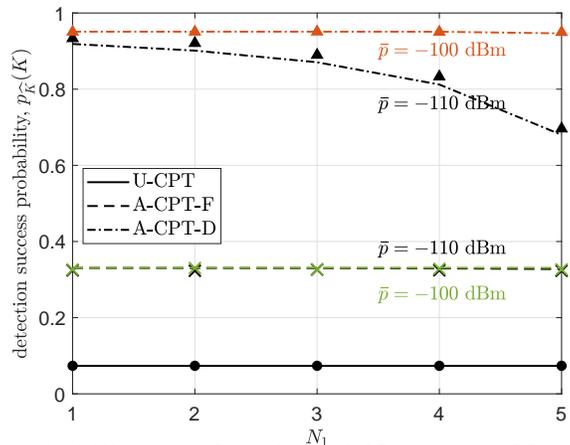}
		\vspace{-6mm}
	\caption{Detection success probability as a function of $N_1$. We set $N_2=N-N_1$. The performance of the U-CPT configuration, which does not uses ($N_1$) broadcast DL symbols, is illustrated only for benchmarking.}
		\vspace{-6mm}
	\label{FigR3}
\end{figure}
Fig.~\ref{FigR3} shows the detection success probability as a function of $N_1$, while $N_2=N-N_1$. Note that U-CPT does not define ($N_1$) DL broadcast symbols, thus, $p_{\widehat{K}_{\text{u-cpt}}}(K)$, obtained merely with $N=6$, does not vary with $N_1$. Interestingly, the performance of A-CPT-F is not significantly affected by the choice of $N_1$.
This is because fading mitigation capabilities of A-CPT-F rely only on channel phase corrections, thus, the performance saturates given a small number of training symbols under favorable DL/UL SNR conditions.
Meanwhile, the situation is quite different in the case of A-CPT-D, where the DL/UL symbol allocation influences significantly the detection success probability as there are more degrees of freedom for resolving the fading-related uncertainties. For the adopted configuration, a small number of DL symbols is preferable when operating with A-CPT-D with $\bar{p}\in\{-140,-130\}$ dB. 
%
%
\begin{figure}[t!]
	\centering
	\includegraphics[width=0.45\textwidth]{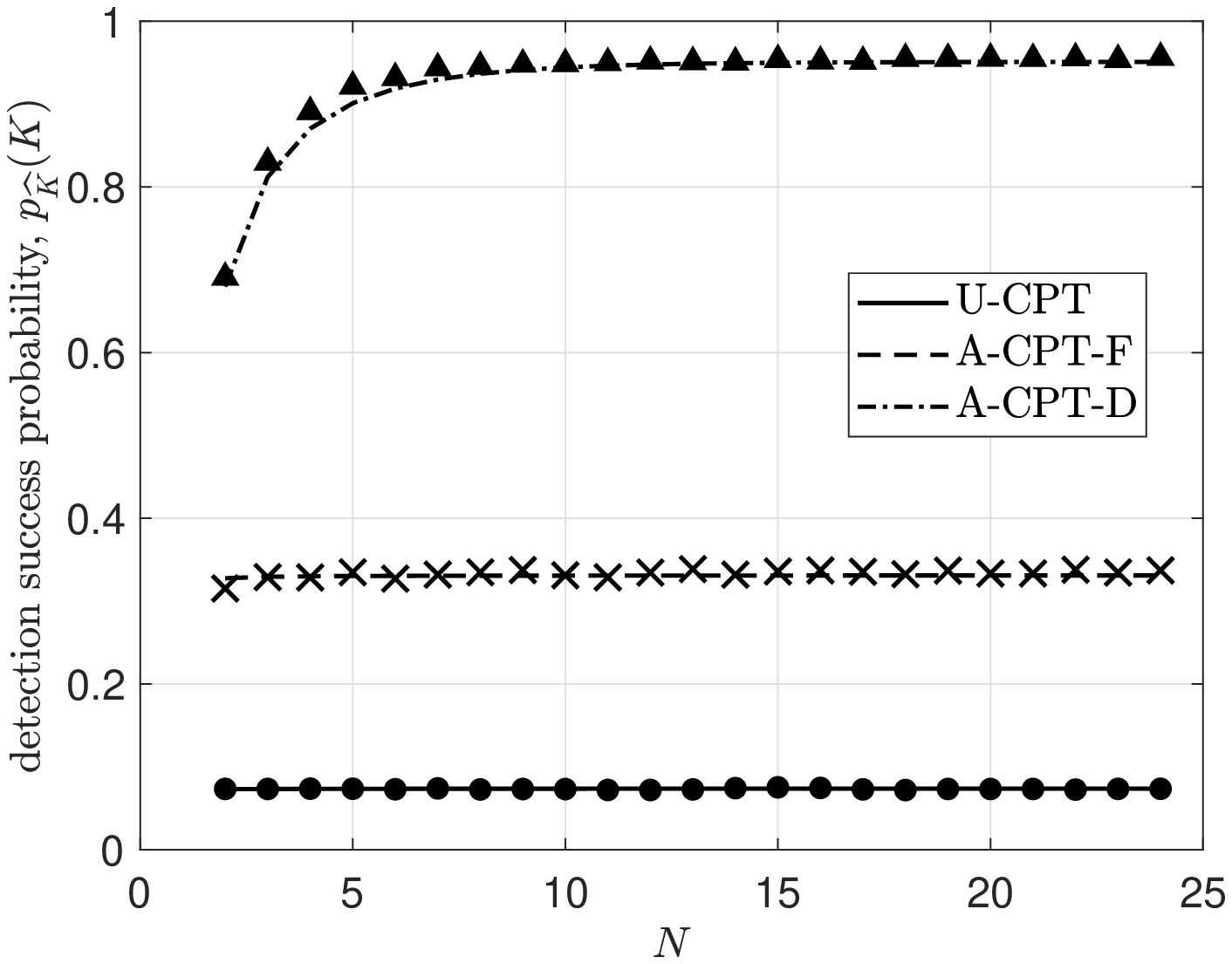}\ \ \ 
	\includegraphics[width=0.45\textwidth]{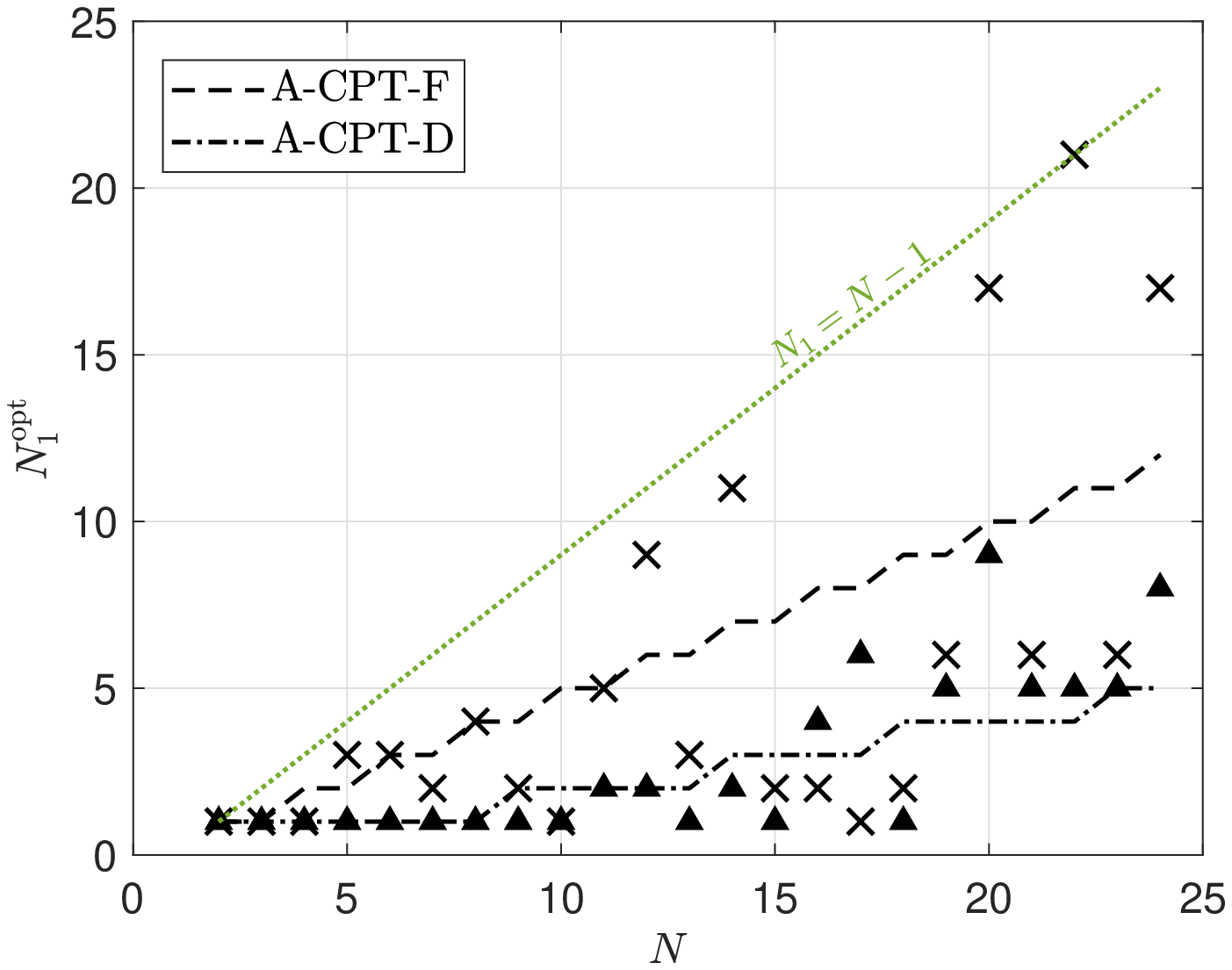}
	\vspace{-4mm}
	\caption{a) Optimum detection success probability (top), and b) optimum $N_1$ for A-CPT mechanisms (bottom) as a function of $N$. Observe that $N_2^\text{opt}=N-N_1^\text{opt}$.}
	\label{FigR4}
	\vspace{-10mm}
\end{figure}

%

The performance impact of increasing the total number of CPT symbols is illustrated in Fig.~\ref{FigR4}a. 
Herein, we test the performance under all possible combinations of $(N_1,N_2)$ with $N_1+N_2=N$ and select the one leading to the best detection success probability, which appears depicted in Fig.~\ref{FigR4}b. Observe that the detection performance under U-CPT and A-CPT-F does not vary significantly with $N$, contrary to what happens with A-CPT-D. This is because SNR conditions are favorable, which means that detection performance is dominated by fading realizations and not the noise. In any case, it seems that configuring a relatively small $N$ guarantees already a near-optimum performance. For instance, the detection success probability with $N=10$ under U-CPT (0.075), A-CPT-F (0.333) and A-CPT-D (0.949), represents already 99.32\%, 99.40\% and 99.37\% of what it would be obtained for $N=24$. 
Finally, notice that A-CPT-F (A-CPT-D) benefits more from DL (UL) rather than UL (DL) CPT symbols, although this ultimately depends on the DL/UL average SNR conditions.
\vspace{-4mm}
\section{Conclusion and Future Research Directions}\label{conclusions}
\vspace{-2mm}
In this work, we introduced a CPT-based framework for detecting the number $K$ of devices that become active, i.e., signal sparsity level,  in a mMTC network under Rayleigh fading. Such information may play a crucial role for sparse signal recovery algorithms aiming at identifying the specific set of active devices. The proposed CPT mechanisms leverage $N$ symbols at the beginning of a transmission block for the purpose of detecting $K$. Specifically, only UL transmissions are exploited when using U-CPT, while A-CPT mechanisms include DL transmissions for CSI estimation that resolve fading uncertainty. We discussed two A-CPT specific implementations: i) A-CPT-F, which implements CSI-based phase corrections while leveraging the same statistical inverse power control used by U-CPT, and ii) A-CPT-D, which implements a dynamic CSI-based inverse power control, although it requires some active devices to remain in silence if their corresponding channels are too faded. For each CPT mechanism, we derived efficient estimators for $K$  and analytically characterized their variance. We showed that the estimator variance increases with $K^2$ and $K$ when operating respectively with U-CPT and A-CPT mechanisms, thus, evincing the superiority of the latter. We derived exact or approximate (semi) closed-form expressions for the PMF of the estimators, and validated their accuracy. Such PMFs were shown to be exponential-, Gaussian- and Student's $t-$like in the case of U-CPT, A-CPT-F and A-CPT-D, respectively.
They allow tractable computation of the detection success probability, thus, becoming valuable for system design/optimization purposes. The numerical results demonstrated the superiority of A-CPT-D under appropriate (but not strict/rigorous) configuration of the probability of the devices to remain silent, $\xi$. The optimum value of $\xi$ was shown to decrease slowly with $K$. Moreover, we analyzed the detection success probability performance under 
two
 different rounding strategies: i) rounding to the nearest integer (NI), and ii) rounding relying on maximum likelihood (ML).
 The results revealed that NI-based rounding offers a performance similar to that of the more sophisticated ML-based rounding.
 We showed that the detection success probability increases with $N$, specially when operating with A-CPT-D, although the performance gain may saturate quickly in high SNR regimes.
%

To conclude, below we enumerate some attractive research directions that are closely related to our contribution in this work and which we would like to pursue in the sequence.
\subsubsection{Exploiting Prior Statistical  Knowledge of $K$}
We have not assumed any statistical  knowledge of $K$. In practice, the coordinator might have some prior expectations based on traffic history, which can be leveraged for more accurate CPT-based detectors. In a future work, we aim to design CPT-based detectors exploiting traffic history. Moreover, note that according to the results illustrated in Fig.~\ref{FigR5}, some information regarding $p_K(k)$ can be also leveraged to properly configure $\xi$, thus, to optimize the detection success probability under A-CPT-D.
\subsubsection{Practical/Arbitrary Channel Models}
In this work, we have limited our discussions to a scenario where channels are subject to quasi-static Rayleigh fading. 
For further studies, either assuming more general/realistic channel models, e.g., Rician fading, or being completely agnostic of channel statistics, are attractive directions. In the case of channels subject to Rician fading, the proposed U-CPT detector needs to be revisited, while it does not work at all under unknown channel statistics. In any case,  proposed A-CPT detectors remain valid and the detection success probability needs to be re-assessed for all the CPT-based detectors. 
%
\subsubsection{CPT Optimized for MIMO Systems}
MIMO technology is key for successful MUD, specially in mMTC networks with sporadic device activations. Therefore, adapting 
our proposed CPT framework to MIMO setups is undoubtedly appealing. Here, the A-CPT detectors require  adjustments to account for multi-antenna CSI acquisition. Due to the potential overhead introduced by multi-antenna CSI training and the limited number of CPT symbols that may be available, an efficient proposal can rely on a compressed CPT training phase that limits the number of communicating antennas and/or  exploits efficiently configured precoders/combiners.
\subsubsection{Joint CPT \& MUD Optimization}
The proposed CPT phase spanning over $N$ symbols and aiming to determine the number of active devices is followed by a MUD phase occupying $N'$ symbols, where the specific set of active devices is detected. Note that the number of active devices detected by CPT works as a prior for MUD mechanisms. An interesting question that we aim to address in a future work is how to efficiently allocate the CPT and MUD symbols given that $N+N'$ is constrained. For this, one needs to jointly assess the performance of both CPT \& MUD phases, which ultimately reveals the (practical) achievable performance of MUD.
\appendices
\section{Proof of Theorem~\ref{EK}}\label{ApK}
Since $f(\mathbf{y};\mathcal{H}_k)\in\mathbb{C}$, we can rewrite \eqref{hK} as
\begin{align}
	\hat{K}
	& =\arg\max_{k} \Big(f\big(\Re\{\mathbf{y}\};\mathcal{H}_k\big)^2 + f\big(\Im\{\mathbf{y}\};\mathcal{H}_k\big)^2\Big),
	\label{hK2}
\end{align}
where 
\begin{align}
	f\big(\Re\{\mathbf{y}\};\mathcal{H}_k\big)=\!\int\! f\big(\Re\{\mathbf{y}\}|h';\mathcal{H}_k\big) f_{H'}(h')\mathrm{d}h',\ \ 
	f\big(\Im\{\mathbf{y}\};\mathcal{H}_k\big)=\!\int\! f\big(\Im\{\mathbf{y}\}|h';\mathcal{H}_k\big) f_{H'}(h')\mathrm{d}h'\nonumber
\end{align}
using the Bayesian approach \cite[Ch.6]{Kay.1998}. Then, we derive $f\big(\Re\{\mathbf{y}\};\mathcal{H}_k\big)$ as shown in \eqref{pyre} at the top of the page.
Therein, (a) is obtained by respectively associating $\Re\{h'\}$ and $\Im\{h'\}$ to integration variables $x_1$ and $x_2$ for notation simplicity, and then exploiting $\Re\{y[n]\}=\sqrt{\bar{p}k}\big(\Re\{h'\}\Re\{s[n]\}-\Im\{h'\}\Im\{s[n]\}\big)+\Re\{w[n]\}$ together with $\Im\{h'\},\Re\{h'\}\sim \mathcal{N}(0,1/2)$ and $\Re\{w[n]\}\in\mathcal{N}(0,1/2)$. Meanwhile (b) comes after some algebraic simplifications that include the sum and expansion of the square of trinomials, and (c) is obtained after integration and using the definition of $\alpha_1(\mathbf{x})$, $\alpha_2(\mathbf{x})$, $\alpha_3(\mathbf{x})$, $\delta_1$ and $\delta_2$. 
Finally, by 
removing the contribution of $2/(\sigma^N\pi^{N/2})$ in \eqref{pyre} as it does not affect the optimization, and leveraging that $f\big(\Im\{\mathbf{y}\};\mathcal{H}_k\big)$ can be also evaluated via \eqref{pyre} but using $\Im\{\mathbf{y}\}$ instead of $\Re\{\mathbf{y}\}$, we obtain \eqref{hKF} after few algebraic operations. \hfill 	\qedsymbol
%
\begin{figure*}
	\small
	\begin{align}
	f\big(\Re\{\mathbf{y}\};\mathcal{H}_k\big)&
	\stackrel{(a)}{=}\!\int\limits_{-\infty}^{\infty}\!\int\limits_{-\infty}^{\infty}\! \mathrm{exp}\bigg(\!\!-\!\frac{1}{\sigma^2}\sum\limits_{n=1}^{N}\Big(\Re\{y[n]\}\!-\!\sqrt{k\bar{p}}\big(x_1\Re\{s[n]\} -\!x_2\Im\{s[n]\}\big)\Big)^2\!\!-\!x_1^{2}\!-\!x_2^{2}\bigg)\frac{\mathrm{d}x_1\mathrm{d}x_2}{\pi^{\frac{N+2}{2}}\!\! \sigma^N}\nonumber\\
	&\stackrel{(b)}{=}\!\int\limits_{-\infty}^{\infty}\int\limits_{-\infty}^{\infty} \mathrm{exp}\bigg(\!-\sum\limits_{n=1}^{N}\frac{\Re\{y[n]\}^2}{\sigma^2}-\frac{2}{\sigma^2}\sum\limits_{n=1}^{N}\Big(\sqrt{k\bar{p}}\Re\{y[n]\}\Im\{s[n]\}x_2-\sqrt{k\bar{p}}\Re\{y[n]\}\Re\{s[n]\}x_1\nonumber\\
	&\qquad\qquad\qquad\qquad\qquad-k\bar{p}\!\ \Re\{s[n]\}\Im\{s[n]\}x_1x_2\Big) -\Big(\frac{Nk\bar{p}}{\sigma^2}+1\Big)\big(x_1^{2}+x_2^{2}\big)\bigg)\frac{\mathrm{d}x_1\mathrm{d}x_2}{\pi^{\frac{N+2}{2}} \sigma^N}\nonumber\\
	&\stackrel{(c)}{=}\frac{2\ \!\mathrm{exp}\Big(\frac{\delta_1\big(\alpha_2(\Re\{\mathbf{y}\})\alpha_3(\Re\{\mathbf{y}\})-\alpha_1(\Re\{\mathbf{y}\})\delta_1\big)+\delta_2\big(\alpha_2(\Re\{\mathbf{y}\})^2+\alpha_3(\Re\{\mathbf{y}\})^2\big)+4\delta_2^2\alpha_1(\Re\{\mathbf{y}\})}{\delta_1^2-4\delta_2^2}\Big)}{\sigma^N\pi^{N/2}\sqrt{4\delta_2^2-\delta_1^2}}\label{pyre}\\
	\bottomrule\nonumber
\end{align}
\vspace{-20mm}
\end{figure*}
\vspace{-4mm}
\section{Proof of Theorem~\ref{the1}}\label{AppT1}
	The probability of the active devices to remain in silence must satisfy $\mathbb{P}(|\hat{h}_i|^2< \mu_i)=\xi$. By exploiting $|\hat{h}_i|^2\sim \mathrm{Exp}(1/\vartheta_i)$, one can easily obtain $\mu_i$ as in \eqref{u1rho}. Then, the common parameter $\rho$ is the solution to $\mathbb{E}\big[\frac{\rho}{|\hat{h}_i|^2}\big||\hat{h}_i|^2\ge \mu_i\big]=\frac{\bar{p}}{\beta_i}$. Note that

\begin{align}
\mathbb{E}\bigg[\frac{\rho}{|\hat{h}_i|^2}\Big||\hat{h}_i|^2\ge\mu_i\bigg]&=\frac{\rho}{\mathbb{P}(|\hat{h}_i|^2\ge  \mu_i)}\int\limits_{\mu_i}^{\infty}\frac{e^{-x/\vartheta_i}}{\vartheta_i x}\mathrm{d}x=-\frac{\rho}{(1-\xi)\vartheta_i}\mathrm{li}(1-\xi),\label{eq}
\end{align}
where last line comes from using $\mathbb{P}(|\hat{h}_i|^2\!\ge \! \mu_i)=1-\xi$, $\mu_i$ in \eqref{u1rho}, and the definition of the logarithmic integral function \cite[eq. (6.2.8)]{Olver.2010}. Finally, one attains $\rho$ in \eqref{u1rho} by equaling \eqref{eq} to $\bar{p}/\beta_i$ and setting $N_1\bar{\gamma}_i\rightarrow \infty,\ \forall i$. \hfill 	\qedsymbol
\vspace{-6mm}
\section{Proof of Theorem~\ref{th2}}\label{AppT2}
%
Note that the PMF of $K''$ conditioned on $K$ can be written as
\begin{align}
p_{K''|K}(k'')&=\frac{\binom{K}{k''}}{\binom{Q}{K}}\!\!\!\!\sum_{\substack{\forall \mathcal{I},\mathcal{J}\subset\mathcal{Q}:\\ \mathcal{I}\cap\mathcal{J}=\emptyset,\\ |\mathcal{I}|=k'',|\mathcal{J}|=K-k''}}\!\!\!\!\prod_{i\in\mathcal{I}}\xi_i\prod_{j\in\mathcal{J}}(1-\xi_j)\label{pK}
\end{align}
$\forall k''=0,\cdots,K$,
where $\xi_i\triangleq \mathbb{P}(|\hat{h}_i|^2<\mu_i),\ \forall i\in\mathcal{Q}$. 
%
Such expression simplifies to
\begin{align}
p_{K''|K}(k'')\!=\!\binom{K}{k''}\xi^{k''}(1\!-\!\xi)^{K-k''}\!\!,\!\ \text{for\ }k''\!=\!0,\cdots,K,\label{pK2}
\end{align}
for scenarios with $\xi_i=\xi,\forall i$. Hereinafter, we proceed with such assumption for tractability. Nevertheless, 
one could still directly work with an arbitrary set $\{\xi_i\}$ via \eqref{pK} at expense of increased complexity.

%

Observe that $K=K'+K''$, thus, \eqref{pK2} can be readily written in terms of $K'$, instead of $K$, as
\begin{align}
p_{K''|K'}(k'')=\binom{K'+k''}{k''}\xi^{k''}(1-\xi)^{K'} \label{pK3}
\end{align}
$\forall k''\!=\!0,\cdots,K_l$. Here, $K_l$ is the solution to $\sum_{k''=0}^{K_l} p_{K''|K'}(k'')=1$, where 
\begin{align}
\sum_{k''=0}^{K_l} p_{K''|K'}(k'')&= \!\sum_{k''=0}^{\infty}\!\! p_{K''|K'}(k'')\!-\!\!\sum_{k''=0}^{\infty}\!\! p_{K''|K'}(k''\!+\!K_l\!+\!1)\nonumber\\
&=\frac{1}{1-\xi}\bigg(1-B_{\xi}\big(1+K_l,1+K'\big)\frac{(1+K_l+K')!}{K_l!K'!}\bigg),\label{Kleq}
\end{align}
which is attained after several algebraic transformations and exploiting the definition and properties of the incomplete beta function \cite[Sec. 8.17]{Olver.2010}.
Since $\sum_{k''=0}^{\infty} p_{K''|K'}(k'')=\frac{1}{1-\xi}$, $K_l$ must be finite. Nevertheless, observe that $K_l\rightarrow\infty$ as $\xi\rightarrow 0$.  By equaling \eqref{Kleq} to 1, one obtains \eqref{Kl}.

Then, by exploiting $\widehat{K'}_\mathrm{a-cpt-d}$, the relaxed MMSE estimate of $K''$ is given by
	\begin{align}
\widehat{K''}_\mathrm{a-cpt-d}^\mathbb{R}&=\sum_{k''=0}^{K_l}k''p_{K''|\widehat{K'}}(k'')=\sum_{k''=0}^{K_l}k''\binom{\widehat{K'}+k''}{k''}\xi^{k''}(1-\xi)^{\widehat{K'}},\label{mmse0}
\end{align}
where last line comes from using \eqref{pK3}. Finally, \eqref{mmse} follows after few algebraic transformations and leveraging the definition of  the (Gaussian) hypergeometric function. 
\hfill 	\qedsymbol
\section{Proof of Theorem~\ref{TH3}}\label{AppTH3}
By departing from \eqref{varKT}, 
one obtains
\begin{align}
\mathrm{var}[\widehat{K}_{\mathrm{a-cpt-d}}^\mathbb{R}]&\!=\!\mathrm{var}\big[\widehat{K'}_\mathrm{a-cpt-d}^\mathbb{R}\big]\!+\!\mathrm{var}\big[\widehat{K''}_\mathrm{a-cpt-d}^\mathbb{R}\big]\!\stackrel{(a)}{=}\!\psi \mathrm{var}\big[\widehat{K'}_\mathrm{a-cpt-d}^\mathbb{R}\big]\!  \stackrel{(b)}{=}\!\psi\Big(\mathrm{var}\Big(\Re\Big\{\frac{\varphi_{\mathcal{K}'}}{\sqrt{\rho}}\Big\}\Big)\!+\!\mathrm{var}\Big(\Re\Big\{\frac{\mathbf{s}^H\mathbf{w}}{N_2\sqrt{\rho}}\Big\}\Big)\Big)\nonumber\\
&\stackrel{(c)}{=}\!\psi\bigg(\sum_{i\in\mathcal{K}'}\mathrm{var}\bigg[\Re\Big\{\frac{\hat{h}_i^*\tilde{h}_i}{|\hat{h}_i|^2}\Big\}\bigg]\!+\!\frac{1}{2N_2\bar{\gamma}_c'}\bigg)\stackrel{(d)}{=}\psi\bigg(\sum_{i\in\mathcal{K}'}\mathrm{var}\bigg[\frac{\Re\{\tilde{h}_i\}}{|\hat{h}_i|}\bigg]\!+\!\frac{1}{2N_2\bar{\gamma}_c'}\bigg),\label{varAPT}
\end{align}
where (a) comes from leveraging the relation between $\widehat{K'}_\mathrm{a-cpt-d}^\mathbb{R}$ and $\widehat{K''}_\mathrm{a-cpt-d}^\mathbb{R}$ given in \eqref{Kp2}, while (b) follows from substituting \eqref{KpE}. Meanwhile, (c) comes from exploiting the definition of $\varphi_{\mathcal{K}'}$ and $\Re\{\mathbf{s}^H\mathbf{w}/(N_2\sqrt{\rho})\}\sim \mathcal{N}(0,1/(2N_2\bar{\gamma}_c'))$. Then, $(d)$ follows after exploiting the fact that $\hat{h}_i^*\tilde{h}_i/|\hat{h}_i|^2$ is equivalently distributed as $ \tilde{h}_i/|\hat{h}_i|$ since $\hat{h}_i^*/|\hat{h}_i|$ is uniformly distributed in the unit circle and independent from $\tilde{h}_i/|\hat{h}_i|$, thus, it does not alter latter's distribution.
%
%
%
%

Observe that $\Re\{\tilde{h}_i\}/|\hat{h}_i|\sim\frac{1}{\sqrt{2N_1\varrho\vartheta_i}}X_i/Y_i$, where
$X_i\sim \mathcal{N}(0,1)$, and 
\begin{align}
f_{Y_i}(y)&=\frac{f_{|\hat{h}_i|}(y)}{1-F_{|\hat{h}_i|}(\sqrt{\mu_i/\vartheta})}=2y e^{\mu_i/\vartheta_i-y^2}\label{fY}
\end{align}
for $y\ge \sqrt{\mu_i/\vartheta_i}$. Let us denote $Z_i\triangleq X_i/Y_i$,
then we can write \eqref{varAPT} as
\begin{align}
\mathrm{var}[\widehat{K}_{\mathrm{a-cpt-d}}^\mathbb{R}]=\frac{\psi}{2}\Big(\sum_{i\in\mathcal{K}'}\frac{\mathrm{var}[Z_i]}{N_1\varrho\vartheta_i}+\frac{1}{N_2\bar{\gamma}_c'}\Big).\label{vaK}
\end{align}
Now, observe that
%
%
\begin{align}
\mathrm{var}[Z_i]&\stackrel{(a)}{=}\mathbb{E}[X^2]\mathbb{E}[Y^{-2}]=\int_{\sqrt{\mu_i/\vartheta_i}}^{\infty}y^{-2}f_Y(y)\mathrm{d}y\stackrel{(b)}{=}2\int_{\sqrt{\mu_i/\vartheta_i}}^{\infty}\frac{1}{y} e^{\frac{\mu_i}{\vartheta_i}-y^2}\mathrm{d}y\stackrel{(c)}{=}e^{\frac{\mu_i}{\vartheta_i}} \mathrm{Ei}\Big(\frac{\mu_i}{\vartheta_i}\Big),\label{varZ} 
\end{align}
where (a) follows since $X$ and $Y$ are independent and $\mathbb{E}[X]=0$, while (b) comes after using \eqref{fY}. Then, we attain (c) by applying simple algebraic transformations and using the definition of the exponential integral \cite[eq. (6.2.1)]{Olver.2010}. Finally, by substituting \eqref{varZ} into \eqref{vaK} while leveraging $\mu_i/\vartheta_i=-\ln(1-\xi)$ (Theorem~\ref{the1}), one attains \eqref{varF}. \hfill 	\qedsymbol
\vspace{-4mm}
\section{Proof of Theorem~\ref{a-cpt-f-T}}\label{TH2}
\vspace{-2mm}
We leverage \eqref{KF}, \eqref{KF2} to attain
\begin{align}
F_{\widehat{K}_{\mathrm{a-cpt-f}|\mathcal{K}}^\mathbb{R}}\!\!(\hat{k})\!=\!\mathbb{P}\big(\widehat{K}_\mathrm{a-cpt-f}^\mathbb{R}\le \hat{k}\big)\!=\!\mathbb{P}\Big( \sum_{i\in\mathcal{K}}U_i\!+\!R \!\le\! \sqrt{2}\hat{k}/\upsilon\Big)\!=\!\int\limits_{0}^{\infty}\!F_R\Big(\frac{\sqrt{2}}{\upsilon}\hat{k}\!-\!\sum_{i\in\mathcal{K}}\!u_i\Big)\!\prod_{i\in\mathcal{K}}\!f_{U_i}(u_i)\mathrm{d}u_i.\label{FKYXi}
\end{align}
Observe that \eqref{FKYXi} is conditional on the specific set of active devices $\mathcal{K}$, and not only on its cardinality $K$. Therefore, different estimation statistics may be observed even if the number of active devices $K$ remains unchanged. Since the set of active devices is unknown to the coordinator, one needs to average all the possible combinations as follows
\begin{align}
F_{\widehat{K}_{\mathrm{a-cpt-f}}|K}^\mathbb{R}(\hat{k})&=\mathbb{E}_{\substack{\forall \mathcal{K}\subseteq \mathcal{Q}\\|\mathcal{K}|=K} }\Big[F_{\widehat{K}_{\mathrm{a-cpt-f}}|\mathcal{K}}^\mathbb{R}(\hat{k})\Big]=\frac{1}{\binom{Q}{K}}\sum_{\substack{\forall \mathcal{K}\subseteq \mathcal{Q}\\|\mathcal{K}|=K}}F_{\widehat{K}_{\mathrm{a-cpt-f}}|\mathcal{K}}^\mathbb{R}(\hat{k}).\label{FKK}
\end{align}
Finally, one attains \eqref{final} after substituting \eqref{FKYXi}, together with the CDF of $R$ and the PDF of $U_i$, into \eqref{FKK}. \hfill 	\qedsymbol
\vspace{-7mm}
\section{Proof of Theorem~\ref{a-cpt-d-T}}\label{TH5}
The CDF of $\widehat{K'}_\mathrm{a-cpt-d}^\mathbb{R}$ can be obtained as follows
\begin{align}
&F_{\widehat{K'}_\mathrm{a-cpt-d}^\mathbb{R}|\mathcal{K}'}(\hat{k'})
\!=\!\mathbb{P}\Big(K'\!+\!\sum_{i\in\mathcal{K}'}Z_i\!+\!V\le \hat{k'}\Big) \!=\!\!\!\underbrace{\int\limits_{-\infty}^\infty\!\!\!\cdots\!\!\!\int\limits_{-\infty}^\infty}_{K'\ \text{integrals}}\!\!\! F_V\Big(\hat{k'}\!-\!K'\!-\!\sum_{i\in\mathcal{K}'}\!z_i\Big)\!\prod_{i\in\mathcal{K}'}\!f_{Z_i}(z_i)\mathrm{d}z_i,\label{KV}
\end{align}
where $F_V(v)=1-\mathrm{erfc}(v\sqrt{N_2\bar{\gamma}_c'})/2$.
However, since such distribution is conditional on the specific set of active and transmitting devices, one must compute
\begin{align}
F_{\widehat{K'}_{\mathrm{a-cpt-d}}|K'}^\mathbb{R}(\hat{k'})&=\mathbb{E}_{\substack{\forall \mathcal{K}'\subseteq \mathcal{Q}\\|\mathcal{K}'|=K'} }\Big[F_{\widehat{K'}_{\mathrm{a-cpt-d}}|\mathcal{K}'}^\mathbb{R}(\hat{k'})\Big]=\frac{1}{\binom{Q}{K'}}\sum_{\substack{\forall \mathcal{K}'\subseteq \mathcal{Q}\\|\mathcal{K}'|=K'}}F_{\widehat{K'}_{\mathrm{a-cpt-d}}|\mathcal{K}'}^\mathbb{R}(\hat{k'}),\label{FKK2}
\end{align}
to bring the condition on $K'$, similar to \eqref{FKK}. This is done by averaging with respect to all the possible sets $\mathcal{K}'$ with fixed cardinality $K'$. Finally, \eqref{KD} is attained by substituting \eqref{KV}, together with $F_V(v)$, into \eqref{FKK2}. \hfill 	\qedsymbol
\vspace{-7mm}
\section{Proof of Theorem~\ref{TH6}}\label{AL3}
Note that $T$ is symmetric around 0 and bell-shaped, thus, a Student's $t$ distribution, which is also more general than a Gaussian distribution, may be a good fit. Several simulation campaigns that we carried out revealed that this is indeed the case. 

At least two moments of $T$ are needed to match those of a scaled Student's $t$ distribution since such distribution is characterized only by scale $s$, and number of degrees of freedom $\nu$. The challenge lies in that $\nu$ must be greater than the moment order, while odd moments cannot be used since they are 0. For instance, this implies that we cannot rely on the first moment, and we cannot fit moments of order higher than 2 in order to allow $\nu\in(2,\infty)$ (which is required for having defined variance as it is the case here, see \eqref{varF}). The latter issue is very important since simulation results evinced that $T$ may accurately fit, in many cases, a scaled Student's $t$ distribution with $\nu$ approaching $2$ from above. To circumvent these issues, we resort to a fitting based on the second moment and characteristic function. 

The second moment of $T$ is given by
\begin{align}
\omega_1\triangleq\mathbb{E}[T^2]=K'\mathbb{E}[Z^2],\label{Z2}
\end{align}
which equals $\omega_1$ in \eqref{omega1omega2} by leveraging the independence and zero-mean features of $\{Z_i\}$ and from using \eqref{varZ} with $\mu_i/\vartheta_i=-\ln(1-\xi)$. Meanwhile, the characteristic function of $T$ can be computed as
\begin{align}
\omega_2(t)&\triangleq \mathbb{E}[e^{\imath tT}]=\mathbb{E}\big[e^{\imath t \sum_{i\in\mathcal{K}'}Z_i}\big]=\mathbb{E}[e^{\imath tZ}]^{K'}\stackrel{(a)}{=}\bigg[\int_{-\infty}^{\infty}\big(\cos (tz)+\imath\sin (tz)\big)f_Z(z)\mathrm{d}z\bigg]^{K'}
\nonumber\\
&\stackrel{(b)}{=}\bigg[2\int_{0}^{\infty}\cos (tz)f_Z(z)\mathrm{d}z\bigg]^{K'}
,\ \forall t\ge 0,\label{C}
\end{align}
where (a) comes from exploiting $e^{\imath a}=\cos a+\imath\sin a$ and expressing the expectation in integral form, while (b) is obtained by leveraging the symmetry of $f_Z(s)$ around $0$, and properties $\cos(-a) =\cos a$, $\sin (-a)=-\sin a$.
\begin{figure}[t!]
	\centering
	\includegraphics[width=0.45\textwidth]{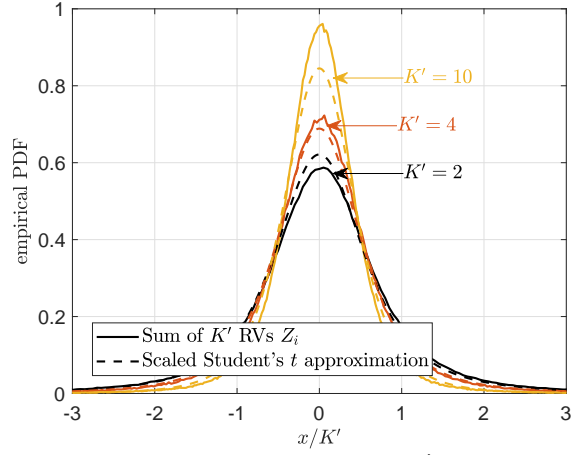}
	\vspace{-5mm}
	\caption{Empirical PDF of the normalized sum of $K'$ i.i.d. RVs distributed as \eqref{fZ2} and corresponding scaled Student's $t$ fitting. We set $\mu/\vartheta=10^{-2}$.}
	\label{figF}
	\vspace{-9mm}
\end{figure}

Now, we proceed to match \eqref{Z2} and \eqref{C} with the second moment and characteristic function 
of a scaled Student's $t$ distribution $s\mathcal{T}(\nu)$, which are respectively given by
\begin{align}
E\big[(s\mathcal{T}(\nu))^2\big]&= \frac{s^2\nu}{\nu-2},\label{E2}\\
\text{CF}\big(s\mathcal{T}(\nu)\big)&=\frac{K_{\nu/2}(\sqrt{\nu}ts)(\sqrt{\nu}ts)^{\nu/2}}{2^{\nu/2-1}\Gamma(\nu/2)},\ \forall t\ge 0.
\label{Eabs}
\end{align}
Then, the system of equations to solve becomes
%
\begin{align}
	\Big\{\omega_1=\frac{s^2\nu}{\nu-2},\qquad
\omega_2(t)=\frac{K_{\nu/2}(\sqrt{\nu}ts)(\sqrt{\nu}ts)^{\nu/2}}{2^{\nu/2-1}\Gamma(\nu/2)}\Big\}.\label{eq2}
\end{align}
Through extensive simulation campaigns, we found that the solution of above system of equations is the most accurate for $t\approx 26$. By setting $t=26$, and combining equations in \eqref{eq2}, one obtains \eqref{eqV}, where $\omega_2$ in \eqref{omega1omega2} matches \eqref{C}. Then, $s$ is attained from $\nu$ by exploiting the first equation in \eqref{eq2}. The accuracy of the fitting is illustrated in Fig.~\ref{figF}.
%
%
%
%
\hfill 	\qedsymbol
\bibliographystyle{IEEEtran}
\bibliography{IEEEabrv,references}
\vspace{-2mm}
\end{document}